\documentclass[11pt]{article}
\usepackage{amsthm}
\usepackage{enumerate}
\usepackage{color}
\usepackage{mathtools}
\usepackage{amsmath}
\usepackage{amssymb}
\usepackage{amsfonts}
\usepackage{siunitx}
\usepackage[utf8]{inputenc} % allow utf-8 input
\usepackage[T1]{fontenc}    % use 8-bit T1 fonts
\usepackage{hyperref}       % hyperlinks
\usepackage{url}            % simple URL typesetting
\usepackage{booktabs}       % professional-quality tables
\usepackage{amsfonts}       % blackboard math symbols
\usepackage{nicefrac}       % compact symbols for 1/2, etc.
\usepackage{microtype}      % microtypography
\usepackage{lipsum}
\usepackage{fancyhdr}       % header
\usepackage{graphicx,graphbox}       % graphics
\usepackage{algorithm}
\usepackage{algpseudocode}
\usepackage{caption}
\usepackage{subcaption}
\usepackage{cite}
\usepackage{soul}
\usepackage{multirow}
\usepackage[normalem]{ulem}
\usepackage{CJKutf8}

\theoremstyle{plain}
\newtheorem{theorem}{Theorem}[section]
\newcommand{\footcomma}{$^{,}$}
\newcommand{\Alex}[1]{{\color{black}#1}}
\newcommand{\Algo}[1]{\table}

\let\Item\item

\begin{document}

\centerline{{\huge Recurrent neural chemical reaction networks that}}

\medskip

\centerline{{\huge approximate arbitrary dynamics}}
 
\medskip
\bigskip

\centerline{
{\renewcommand{\thefootnote}{*}
\large Alexander Dack\footnote[1]{\label{authours}
Email of co-corresponding authours: alex.dack14@imperial.ac.uk or  tp525@cam.ac.uk. 
}\footcomma
\renewcommand{\thefootnote}{\arabic{footnote}}\footnote[1]{
  Department of Bioengineering and Imperial College Centre for Synthetic Biology, Imperial College London, Exhibition Road, London, SW7 2AZ, UK.}
\qquad
Benjamin Qureshi\footnotemark[1]{}
\qquad
Thomas E. Ouldridge\footnotemark[1]{}
\qquad 
Tomislav Plesa\renewcommand{\thefootnote}{*}\footnotemark[1]\footcomma\renewcommand{\thefootnote}{\arabic{footnote}}\footnote[2]{
Department of Applied Mathematics and Theoretical Physics, University of Cambridge, Centre for Mathematical Sciences, Wilberforce Road, Cambridge, CB3 0WA, UK.
}
}}
\medskip
\bigskip

\noindent
{\bf Keywords}: Chemical Reaction Networks; 
\; Artificial Neural Networks;
\; Dynamical Systems
\medskip
\bigskip

\noindent
{\bf Abstract}:
Many important phenomena 
in biochemistry and biology exploit dynamical features
such as multi-stability, oscillations, and chaos.
Construction of novel chemical systems 
with such rich dynamics is a challenging problem
central to the fields of synthetic biology and molecular nanotechnology.
In this paper, we address this problem by 
putting forward a molecular version of 
a recurrent artificial neural network, which we call 
\emph{recurrent neural chemical reaction network} (RNCRN). 
The RNCRN uses a modular architecture -- 
a network of chemical neurons -- to approximate arbitrary dynamics.
We first prove that with sufficiently many
 chemical neurons and suitably fast reactions, the RNCRN
can be systematically trained to achieve any dynamics. 
RNCRNs with relatively small number of 
chemical neurons and a moderate range of reaction rates
are then trained to display a variety of 
biologically-important dynamical features.
We also demonstrate that such RNCRNs are
experimentally implementable with 
DNA-strand-displacement technologies.
\noindent

\section{Introduction}
\emph{Artificial neural networks} (ANNs) are a set of algorithms, 
inspired by the structure and function of the brain, 
that are commonly implemented on electronic machines~\cite{mcculloch43a}. 
Owing to their powerful function approximating abilities~\cite{cybenko_approximation_1989, hornik_multilayer_1989,pinkus_approximation_1999},
ANNs have been used to solve a range of 
data-driven problems involving 
identification, prediction and classification
of patterns, e.g. in text, speech and image analysis~\cite{lecun_deep_2015}.
These networks, and particularly their recurrent counterparts~\cite{chen_neural_2018, hopfield_neurons_1984,krotov_large_2021},
have also been embedded into dynamical systems~\cite{perko_ODEs}
to learn arbitrary dynamical behaviors~\cite{ funahashi_approximation_1993,pathak_model-free_2018},
including oscillations~\cite{kimura_learning_1998} 
and chaos~\cite{sato1991learning,bucci_control_2019}.

ANNs implemented on electronic machines
can be challenging to interface with chemical 
and biological systems. To address this issue, 
as done with some other chemically-realized systems~\cite{cardelli_electric_2020, arkin_computational_1994, seelig_enzyme-free_2006, qian_scaling_2011, del_vecchio_control_2016,briat_antithetic_2016, plesa_quasi-robust_2021},
there has been substantial 
interest in directly embedding ANNs into 
\emph{chemical reaction networks} (CRNs)
- a mathematical framework used for modelling (bio)chemical processes~\cite{feinberg_CRNs_1979}.
In this paper, such CRNs, that execute ANN algorithms, 
are said to be \emph{neural}.
Current literature has focused on designing neural CRNs to solve various 
pattern-recognition problems~\cite{hjelmfelt_chemical_1991, qian_neural_2011,kim_neural_2004, chiang_reconfigurable_2015, vasic_programming_2022, linder_et_al, anderson_reaction_2021, nagipogu_neuralcrns:_2024,kieffer_molecular_2023, Samaniego2024, van_der_linden_dna_2022, brijder_chemical_2017, cherry_scaling_2018, xiong_molecular_2022,
moorman_dynamical_2019,lopez_molecular_2018},
with applications to e.g. virus detection~\cite{vasic_programming_2022}
and disease diagnostics~\cite{lopez_molecular_2018},
with successful experimental implementations using 
substrates such as nucleic acids and microorganisms~\cite{genot_scaling_2013, lopez_molecular_2018, pandi_metabolic_2019, li_synthetic_2021, okumura_nonlinear_2022}.
Particular focus has been placed on both 
theoretical investigations~\cite{moorman_dynamical_2019, anderson_reaction_2021, okumura_nonlinear_2022,vasic_programming_2022} and 
experimental implementations~\cite{okumura_nonlinear_2022,xiong_molecular_2022} 
of neural CRNs based on the classical architecture consisting of 
layers of perceptrons
- processing units whose outputs are non-linear functions of weighted sum of their inputs~\cite{rosenblatt1957perceptron}.

The neural CRNs put forward in the literature
for molecular pattern-recognition problems
take an input set of chemical concentrations 
and produce in the long-run 
\emph{static} (equilibrium) concentrations as the output.
As such, these CRNs cannot in general
be trained to produce a predefined \emph{dynamical} 
(non-equilibrium) output.
However, it is precisely the dynamical signals
that can communicate important messages in biology.
In particular, the time-evolution of 
(bio)chemical species concentrations can be modelled as a 
dynamical system~\cite{feinberg_CRNs_1979,perko_ODEs},
and the detailed information 
enconded in this time-evolution can play important roles  in biology~\cite{tkacik_information_2025}.
For example, fundamental processes 
such as cellular differentiation,
circadian clocks and DNA repair have been modelled 
as dynamical systems that exhibit features such as coexistence of multiple 
favorable states (multi-stability) and oscillations~\cite{xiong_positive-feedback-based_2003, hardin_feedback_1990, lev_bar-or_generation_2000}.
Furthermore, chaotic dynamical systems
have been used to explain survival advantage of cellular populations~\cite{heltberg_chaotic_2019}, 
and heterogeneity in bacterial responses to stress~\cite{choudhary_chaos_2023}.
Aside from occurring in native (bio)chemical systems, 
such dynamical features have also been encoded
in a variety of experimentally-built synthetic ones. 
Examples include DNA-strand-displacement systems 
that display oscillations~\cite{srinivas_enzyme-free_2017} 
and process temporal information~\cite{lapteva_dna_2022},
enzyme-aided DNA systems displaying oscillations, chaos, spiking, 
and molecular event recording~\cite{fujii_predatorprey_2013, kishi_programmable_2018, lobato-dauzier_neural_2024}, 
synthetic gene-regulatory networks displaying oscillations~\cite{stricker_fast_2008},
all the way to multi-stable synthetic bacteria~\cite{zhu_synthetic_2022}.

One class of neural systems, which can be suitable for dynamical 
problems, are \emph{chemical reservoir computers}~\cite{soloveichik_dna_2013, baltussen_chemical_2024}. These chemical-electronic hybrid systems 
consist of a suitable CRN (called the chemical reservoir)
which is coupled to a single layer implemented on an electronic computer. The authors from~\cite{baltussen_chemical_2024}
choose a particular CRN with rich dynamics, 
and demonstrate that the chemical reservoir computer
can then solve a number of computational problems, 
including emulating and predicting dynamical systems.
To achieve these tasks, only a single layer
is trained, which is less computationally intensive 
than the classical training of ANNs.
However, these chemical reservoir computers have not been shown
to possess the ability to replicate arbitrary dynamical behaviors
and, furthermore, they are not purely chemical, relying on an 
electronic computer.

To bridge the gap,  in this paper we introduce 
a family of novel (purely chemical) neural CRNs 
that can replicate arbitrary dynamical behaviors.
These CRNs embed suitable perceptron-based recurrent ANNs, 
and are called \emph{recurrent neural chemical reaction networks} (RNCRNs).
The RNCRNs contain two types of chemical species:
the \emph{executive} species, that execute the desired dynamics, 
and chemical analogues of perceptrons~\cite{anderson_reaction_2021},
called the \emph{chemical perceptrons}, 
that fine-tune the dynamics of the executive species in 
a recurrent-neural-network fashion;
see Figure~\ref{fig:RNCRN} for a schematic representation.
By exploiting the universal approximation abilities of ANNs, 
we prove that RNCRNs can be systematically trained to 
achieve dynamics of any well-behaved target dynamical system,
provided that there are sufficiently many chemical perceptrons, 
and provided that these species are governed by sufficiently 
fast reactions. We show that non-trivial dynamical behaviors can be replicated with relatively few and only moderately fast chemical perceptrons. Furthermore, we demonstrate that RNCRN-based
systems can be implemented under experimentally-feasible conditions 
using DNA-strand-displacement technologies.

 \begin{figure}[ht]
\centering
\includegraphics[width=\columnwidth, trim={4.6cm 2cm 4.6cm 1.8cm},clip]{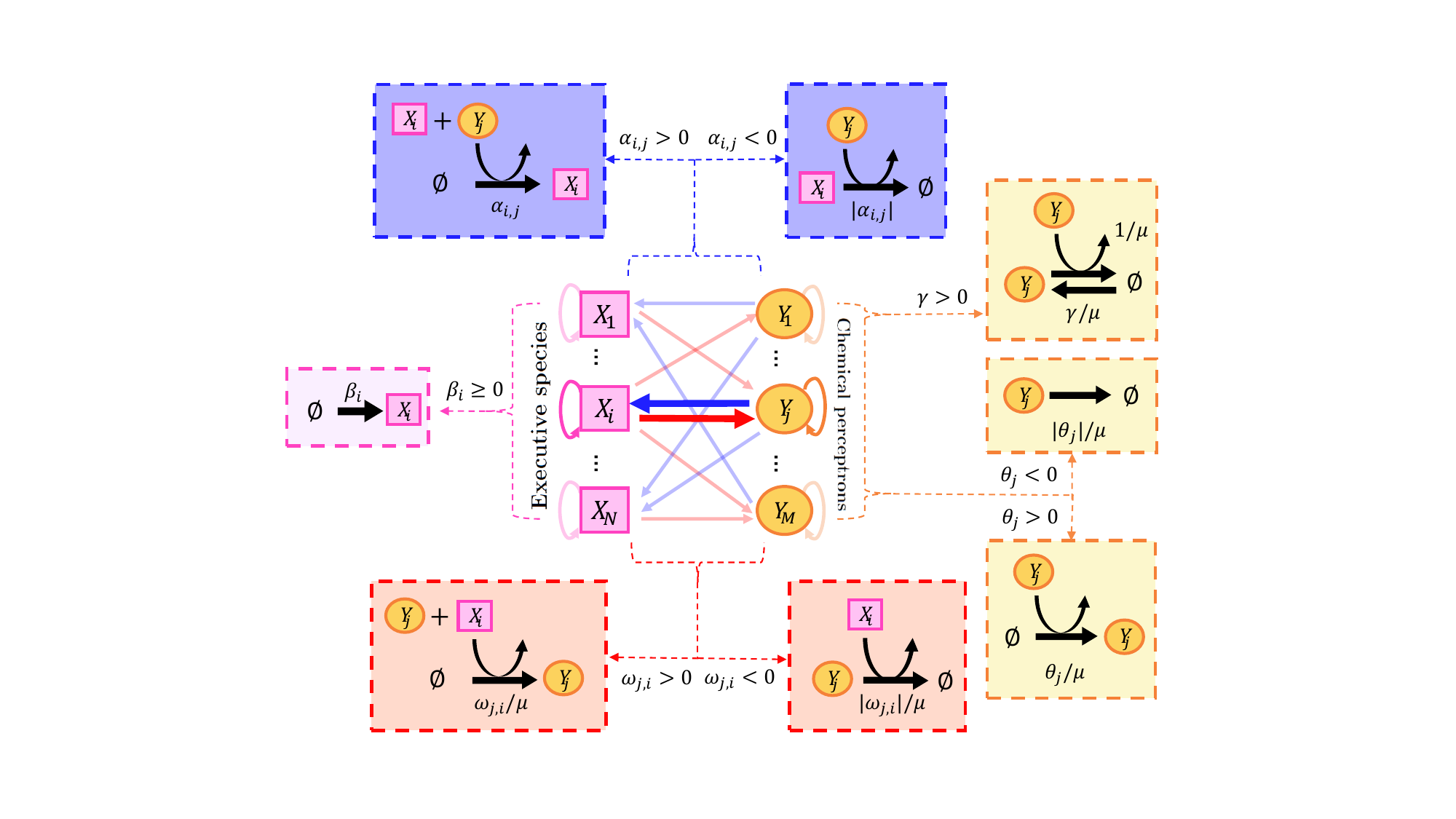}
\caption{A visualization of an RNCRN with executive species $X_1, X_2, \dots, X_N$, and a single layer of chemical perceptrons $Y_1, Y_2, \dots, Y_M$. The interaction between the executive species and chemical perceptrons is outlined in the center via arrows, while the surrounding 
boxes display more details. In particular, the purple box
on the left displays chemical reactions that involve only the executive species; similarly, the yellow boxes on the right show reactions only involving the chemical perceptrons; here, black curved arrows
indicate catalysis. On the other hand, 
the boxes on the top and bottom display reactions that involve
both types of species; these interactions are catalytic in nature. 
In particular, shown in blue on the top are the reactions affecting the executive species, where the chemical perceptrons are catalysts; on the other hand, the orange boxes on the bottom display reactions which affect the chemical perceptrons, with executive species being catalysts. The corresponding chemical reaction network is given by~(\ref{eq:CRN_single_layer}) in Appendix~\ref{app:single_layer_RNCRN};
see also Table~\ref{table:example_CRN}.}
\label{fig:RNCRN}
\end{figure}

The paper is organized as follows. 
In Section~\ref{sec:models_and_methods}, 
we present background theory on CRNs and ANNs. 
In Section~\ref{sec:results}, we introduce RNCRNs, 
and outline their universal approximation abilities,
which are rigorously proved in Appendix~\ref{app:single_layer_RNCRN}.
In Section~\ref{sec:examples}, we present an algorithm 
for training RNCRNs as Algorithm~\ref{algh:RNCRN},
which we apply to obtain relatively simple RNCRNs
displaying multi-stability, oscillations, and chaos;
more details about the examples can be found in Appendix~\ref{app:code_and_data},
while the pseudocode for Algorithm~\ref{algh:RNCRN} in Appendix~\ref{sec:pseudocode}.
In addition, in Appendices~\ref{sec:implementation}
and~\ref{sec:robustness}, we study some 
experimentally-important properties 
of some of the example systems. 
Finally, we provide a summary and discussion in Section~\ref{sec:con_disc}; 
some further details can be found 
in Appendix~\ref{sec:compare}.

\section{Background theory} 
\label{sec:models_and_methods}
In this section, we present some theory 
of chemical reaction networks and artificial neural networks.
Throughout the paper, we assume that all
the variables are dimensionless, unless stated otherwise.

\subsection{Chemical reaction networks (CRNs)}
\label{sec:crns_background} 
Consider a system of two ordinary-differential equations (ODEs) given by
\begin{align}
\frac{\mathrm{d} x}{\mathrm{d} t} 
& = \beta + \alpha x y,
\; \; \; \; \; \; 
\frac{\mathrm{d} y}{\mathrm{d} t} 
= \gamma + \theta y + \omega x y  - y^2. 
\label{eq:simple_set_ode_1}
\end{align}
Let us assume that the real (dimensionless) parameters 
$\beta, \gamma > 0$ are positive, 
while $\alpha, \theta, \omega \in \mathbb{R}$
can be either positive or negative.
Then, dynamical system~(\ref{eq:simple_set_ode_1})
is an example of \emph{reaction-rate equations} (RREs)
- ODEs that model time-evolution of concentrations of chemical species
under suitable conditions~\cite{feinberg_CRNs_1979}.
In particular, $x = x(t) \ge 0$ and $y = y(t) \ge 0$
satisfying RREs~(\ref{eq:simple_set_ode_1})
can be interpreted as (dimensionless) concentrations
of some chemical species $X$ 
and $Y$ at time $t \ge 0$, respectively.

\textbf{Chemical reactions}. Each term on the right-hand 
side of the RREs~(\ref{eq:simple_set_ode_1}) 
can be interpreted as the rate of a chemical reaction, 
which we summarize in Table~\ref{table:example_CRN}. 
For example, the term $\beta > 0$ from~(\ref{eq:simple_set_ode_1})
can be interpreted as the rate of the reaction 
$\varnothing \xrightarrow[]{\beta} X$, 
which describes a production of species $X$ 
from some species which we do not explicitly model, 
and denote by $\varnothing$. 
On the other hand, no chemical reaction can be
associated to the term $\beta < 0$~\cite{plesa_chemical_2016}.
This information is summarized in the first row of 
Table~\ref{table:example_CRN}. Similarly, in the second
row and second column, it is shown that the term $\alpha x y$ from~(\ref{eq:simple_set_ode_1}) 
can be interpreted as the rate of 
$X + Y \xrightarrow[]{\alpha} 2 X + Y$ 
if $\alpha > 0$. According to this reaction, when species 
$X$ and $Y$ react, one molecule of $X$ is produced.
Since then the molecular count of 
$Y$ remains unchanged, we say that $Y$ is 
a \emph{catalyst} in the reaction.
On the other hand, as presented in the second row and third column
of Table~\ref{table:example_CRN},
when $\alpha < 0$, the term $\alpha x y$ 
can be interpreted as the rate of  $X + Y \xrightarrow[]{|\alpha|} Y$, 
where $|\alpha|$ denotes the absolute value of $\alpha$. 
According to this reaction, when $X$ and $Y$ react, 
one molecule of $X$ is degraded with $Y$ being a catalyst. 

\begin{table}[ht] 
\caption{Chemical reactions induced by RREs~(\ref{eq:simple_set_ode_1}).}
\centering 
\begin{tabular}{c|c|c} 
\hline\hline
\textbf{Term in RREs}~(\ref{eq:simple_set_ode_1}) 
& \textbf{Reaction for positive term}
& \textbf{Reaction for negative term} \\ [0.5ex]
\hline & \\ [-2.0ex] 
$\beta$ & $\varnothing \xrightarrow[]{\beta} X$ 
& Never negative \\ 
\hline
$\alpha x y$ & $X + Y \xrightarrow[]{\alpha} 2 X + Y$ 
& $X + Y \xrightarrow[]{|\alpha|} Y$\\ 
\hline
$\gamma$ & $\varnothing \xrightarrow[]{\gamma} Y$ 
& Never negative \\ 
\hline
$\theta y$ & $Y \xrightarrow[]{\theta} 2 Y$
&  $Y \xrightarrow[]{|\theta|} \varnothing$  \\ 
\hline
$\omega x y$ & $X + Y \xrightarrow[]{\omega} X + 2 Y$ 
& $X + Y \xrightarrow[]{|\omega|} X$   \\ 
\hline
$-y^2$ & Never positive & 
$2 Y \xrightarrow[]{1} Y$ \\ [0.2ex]
\hline\hline
\end{tabular}
\label{table:example_CRN}
\end{table}

The correspondence between RREs~(\ref{eq:simple_set_ode_1}) 
and the chemical reactions, shown in Table~\ref{table:example_CRN},
is called \emph{mass-action kinetics}~\cite{feinberg_CRNs_1979}:
the rate of a reaction is given by the product of the concentrations
of the reactants (species on the left-hand side of the reaction)
multiplied by the \emph{rate coefficient} - 
the positive number displayed above the reaction arrow.
Let us note that the rate coefficient in
$- y^2$ from~(\ref{eq:simple_set_ode_1}) is fixed to $1$,
and this term is never positive, 
as summarized in the final row of Table~\ref{table:example_CRN}.
Systems of chemical reactions, such as
those presented in Table~\ref{table:example_CRN}, 
are called \emph{chemical reaction networks} (CRNs)
under mass-action kinetics~\cite{feinberg_CRNs_1979}.
In this paper, we denote the chemical species by 
$X_1,X_2,\ldots, Y_1, Y_2,\ldots$, 
their concentrations by respectively 
$x_1,x_2,\ldots, y_1, y_2,\ldots$,
and the rate coefficients of the underlying reactions
by $\alpha,\beta,\gamma, \theta, \omega$
with appropriate subscripts;
additional symbols are introduced when we consider
experimental implementations of CRNs.

\textbf{Experimental implementation}. 
In this paper, we focus on the mass-action chemical reactions 
of the form given in Table~\ref{table:example_CRN}.
All of these reactions have at most two reactants;
consequently, the induced RREs have quadratic polynomials 
on the right-hand side, as in~(\ref{eq:simple_set_ode_1}).
However, let us stress that the abstract reactions 
from Table~\ref{table:example_CRN} 
cannot be experimentally implemented in their original form;
for example, $\varnothing \xrightarrow[]{} X_1$ 
cannot be implemented as it contains a non-specific species $\varnothing$,
while $X + Y \xrightarrow[]{} 2 X + Y$
cannot be implemented as it describes a catalytic production 
as a single chemical step.
Nevertheless, CRNs consisting of any number of 
the abstract reactions from Table~\ref{table:example_CRN}, 
including reactions that have the same reactants, but different products,
such as $X + Y \xrightarrow[]{} 2 X + Y$ 
and $X + Y \xrightarrow[]{} Y$,
can be approximated by experimentally implementable CRNs~\cite{wilhelm_chemical_2000, plesa_stochastic_2023,soloveichik_dna_2010, chen_programmable_2013,thachuk_implementing_2019}. 
This approximation is achieved by introducing into the abstract CRNs
suitable additional chemical species and reactions.
In particular, one possible implementation
is via DNA strand-displacement \cite{soloveichik_dna_2010, chen_programmable_2013}, 
or enzyme-aided DNA systems~\cite{thachuk_implementing_2019}. 
For more details, see Section~\ref{sec:multi_stable} 
and Appendix~\ref{sec:implementation}.

\subsection{Artificial neural networks (ANNs)}
Artificial neural networks (ANNs) are systems of connected
processing units called artificial neurons.
In this paper, we consider a particular 
type of artificial neurons called 
the \emph{perceptron}~\cite{rosenblatt1957perceptron}. 
Given a set of inputs, a perceptron first applies an affine
function, followed by a suitable non-linear function, 
called the activation function, to produce a single output.
More precisely, reusing some of the symbols 
from Section~\ref{sec:crns_background}, 
let $x_1,x_2,\ldots,x_N \in \mathbb{R}$
be the input values, $\omega_{1},\omega_{2},
\ldots,\omega_{N} \in \mathbb{R}$ the weights,
$\theta\in \mathbb{R}$ a bias, and 
$\sigma : \mathbb{R} \to \mathbb{R}$ a suitable
non-linear function. Then, the perceptron is a function
$y : \mathbb{R}^N \to \mathbb{R}$ defined as
\begin{align}
y(x_1,\ldots,x_N) &
=\sigma \left(\sum_{j=1}^N \omega_{j} x_j + \theta \right).
\label{eq:perceptron}
\end{align}

\textbf{Chemical perceptron}. 
A natural question arises: is there a CRN with single species 
$Y$ such that its RRE has a unique stable equilibrium 
of the form~(\ref{eq:perceptron})? 
Such an RRE has been put forward and analysed
in~\cite{anderson_reaction_2021}, and is given by
\begin{align}
\frac{\mathrm{d} y}{\mathrm{d}t} & = 
\gamma + \left(\sum_{j=1}^N \omega_{j} x_j + \theta \right) y - y^2,
\label{eq:smoothmax_perceptron_ode}
\end{align}
where $x_1,x_2,\ldots,x_N$ are parameters.
Since $\sum_{j=1}^N \omega_{j} x_j + \theta$ has a fixed sign,
one can readily use Table~\ref{table:example_CRN} to
write down a CRN corresponding to~(\ref{eq:smoothmax_perceptron_ode}).
We call the species $Y$ with concentration $y$ satisfying~(\ref{eq:smoothmax_perceptron_ode}) a \emph{chemical perceptron}. 
Setting the left-hand side in~(\ref{eq:smoothmax_perceptron_ode})
to zero, one finds that the perceptron concentration 
has a globally stable equilibrium $y^*$, given by
\begin{align}
    y^* & = \sigma_{\gamma}\left(\sum_{j=1}^N \omega_{j} x_j + \theta \right)
    \equiv \frac{1}{2}\left[
    \left(\sum_{j=1}^N \omega_{j} x_j +\theta \right) + \sqrt{\left( \sum_{j=1}^N\omega_{j} x_j + \theta\right)^2 + 4\gamma} \right]. 
    \label{eq:smu_activation}
\end{align}
We call $\sigma_{\gamma}$  with $\gamma >0$
a \emph{chemical activation function}. In this paper, we allow any $\gamma > 0$, 
but note that $\sigma_{\gamma}$ approaches the rectified linear unit (ReLU) activation function in the special case as $\gamma \to 0$~\cite{anderson_reaction_2021},
which is a common choice in the ANN literature~\cite{lecun_deep_2015}.

\section{Recurrent neural chemical reaction networks (RNCRNs)} 
\label{sec:results}
Let us consider a \emph{target} ODE system
with initial conditions, given by
\begin{align}
\frac{\mathrm{d} \bar{x}_i}{\mathrm{d} t} & = f_i(\bar{x}_1,\ldots,\bar{x}_N),  
\hspace{0.5cm} \bar{x}_i(0) = a_i \ge 0, 
\; \; \; \textrm{for } i = 1, 2, \ldots, N.
\label{eq:target_ODEs}
\end{align}
In what follows, we call the ODE right-hand side
$(f_1,f_2,\ldots,f_N)$ a \emph{vector field}, 
which we assume is sufficiently smooth.
Without loss of generality, we assume 
that~(\ref{eq:target_ODEs}) has desirable dynamical 
features in the positive orthant $\mathbb{R}_{>}^N$.
If such features are located elsewhere, 
a suitable affine change of coordinates can be 
used to move these features to the positive orthant~\cite{plesa_chemical_2016}. 

We wish to find a neural CRN with some chemical species 
$X_1,\dots,X_N$ whose concentrations approximate the
solutions of the target ODEs~(\ref{eq:target_ODEs}). 
Inspired by the chemical perceptron~(\ref{eq:smoothmax_perceptron_ode}), 
let us consider the RREs and initial conditions given by
\begin{align}
\frac{\mathrm{d} x_i}{\mathrm{d} t} 
& = \beta_{i} +  x_i \sum_{j=1}^M \alpha_{i,j} y_j,  && x_i(0) = a_i, 
\; \; \; \textrm{for } i = 1, 2, \ldots, N, \nonumber \\
\frac{\mathrm{d} y_j}{\mathrm{d} t} & = \frac{\gamma}{\mu} 
+ \frac{\theta_{j}}{\mu} y_j
+ \left(\sum_{i=1}^N \frac{\omega_{j,i}}{\mu} x_i\right) y_j 
- \frac{1}{\mu} y_j^2, 
&& y_j(0) = b_j, 
\; \; \; \textrm{for } j = 1, 2, \ldots, M.
\label{eq:single_layer_RRE}
\end{align}
We call the CRN corresponding to the 
RREs~(\ref{eq:single_layer_RRE}) a 
\emph{recurrent neural chemical reaction network} (RNCRN),
and display it schematically in Figure~\ref{fig:RNCRN}. 
The RNCRN consists of two sub-networks:
the executive system and the neural system.
The executive system contains chemical reactions which directly change 
the \emph{executive} species $X_1,X_2,\ldots,X_N$.
Note that the initial conditions for the executive species
from~(\ref{eq:single_layer_RRE}) match the 
target initial conditions from~(\ref{eq:target_ODEs}).
On the other hand, the neural system contains 
the reactions which directly influence 
the auxiliary species $Y_1,Y_2,\ldots,Y_M$, 
for which we allow arbitrary initial conditions
$b_1,b_2,\ldots,b_M \ge 0$. These species
can be formally identified as chemical perceptrons, 
see~(\ref{eq:smoothmax_perceptron_ode}).
However, let us stress that the RRE~(\ref{eq:smoothmax_perceptron_ode}) depends on the \emph{parameters} $x_1,x_2,\ldots,x_N$.
In contrast, the concentrations of chemical perceptrons from the RNCRN
depend on the time-dependent executive \emph{variables} $x_1(t),x_2(t),\ldots,x_N(t)$,
which in turn depend on the perceptron concentrations 
$y_1(t),y_2(t),\ldots,y_M(t)$.
In other words, there is a feedback between the executive
and neural systems, giving the RNCRN a recurrent character.
This feedback is catalytic in nature: the chemical perceptrons 
are catalysts in the executive system;
similarly, executive species are catalysts in the neural system.
 
\textbf{Main result: Universal approximation.} 
We now wish to choose the parameters in the RNCRN
so that the concentrations of the executive species 
$x_i(t)$ from~(\ref{eq:single_layer_RRE}) are close to the 
target variables $\bar{x}_i(t)$ from~(\ref{eq:target_ODEs}).
Key to achieving this match is the parameter $\mu > 0$,
which sets the speed at which the perceptrons equilibrate 
relative to the executive system. 
This speed can be formally set to be infinite by 
multiplying the RREs for the perceptrons
in~(\ref{eq:single_layer_RRE}) by $\mu$
and then fixing $\mu = 0$. These RREs 
then become algebraic equations, whose solutions 
are given by $y_j^* = \sigma_{\gamma} 
(\sum_{i=1}^N \omega_{j,i} x_i + \theta_{j})$, 
where $\sigma_{\gamma}$ is of the form~(\ref{eq:smu_activation}).
We then say that the chemical perceptrons 
are in the \emph{quasi-static} state.
In this case, the executive species 
are governed by the \emph{reduced} ODEs, given by
\begin{align}
\frac{\mathrm{d} \tilde{x}_i}{\mathrm{d} t} & = 
g_i(\tilde{x}_1,\ldots,\tilde{x}_N) = \beta_i + \tilde{x}_i \sum_{j=1}^M \alpha_{i,j} 
\sigma_{\gamma} \left(\sum_{k=1}^N \omega_{j,k} \tilde{x}_k + \theta_{j} \right),  
\;\;  \tilde{x}_i(0) = a_i,
\; \; \textrm{for } i = 1, 2, \ldots, N.
\label{eq:single_layer_RRE_reduced}
\end{align}
This reduced system allows us to prove that the RNCRN can in principle be 
fine-tuned to execute the target dynamics arbitrarily closely.
In particular, to achieve this task, we follow two steps.
 
Firstly, we assume that the chemical perceptrons are in the quasi-static state,
i.e. we consider the reduced system~(\ref{eq:single_layer_RRE_reduced}). 
Using the classical (static) theory from ANNs, 
it follows that the rate coefficients from the RNCRN
can be fine-tuned so that the vector field from the 
reduced system~(\ref{eq:single_layer_RRE_reduced})
is close to that of the target system~(\ref{eq:target_ODEs}).
To ensure a good vector field match, 
one in general requires sufficiently many chemical perceptrons.
We call this first step the \emph{quasi-static approximation}. 

Secondly, we disregard the assumption that the chemical perceptrons
are in the quasi-static state, i.e. we consider 
the full system~(\ref{eq:single_layer_RRE}).
Nevertheless, we substitute into the full system
the rate coefficients found in the first step.
Using perturbation theory, it follows that,
under this parameter choice, the concentrations of the executive species 
from the full system~(\ref{eq:single_layer_RRE}) approximate the 
dependent variables from the target system~(\ref{eq:target_ODEs}) arbitrarily closely, 
provided that the chemical perceptrons fire sufficiently (but finitely) fast. 
We call this second step the \emph{dynamical approximation}. 

In summary, the RNCRN induced by~(\ref{eq:single_layer_RRE}) 
with sufficiently many chemical perceptrons ($M \ge 1$ large enough)
which act sufficiently fast ($\mu > 0$ small enough)
can execute any target dynamics.
This universal approximation result is stated
rigorously and proved in 
Appendix~\ref{app:single_layer_RNCRN}.

\section{Examples}
\label{sec:examples}
In Section~\ref{sec:results}, we have outlined 
a two-step procedure used to prove that 
 RNCRNs can theoretically execute any desired dynamics. 
Aside from being of theoretical value, these two steps
also form a basis for a practical method to train RNCRNs, 
which is presented as Algorithm~\ref{algh:RNCRN}.
In this section, we show that Algorithm~\ref{algh:RNCRN}
can be used to train RNCRNs with a relatively small number of 
perceptrons $M$ and moderate perceptron speed $\mu$ 
to achieve predefined multi-stability, oscillations, and chaos.

\setcounter{table}{1}
\begin{table}[h]
\renewcommand\tablename{Algorithm}\setcounter{table}{0}
\hrule
\vskip 2.5 mm
Fix a target system~(\ref{eq:target_ODEs})
and target compact sets 
$\mathbb{K}_1, \mathbb{K}_2, \ldots, \mathbb{K}_N \subset (0,+\infty)$.
Fix also the rate coefficients $\beta_1,\beta_2,\ldots,\beta_N \ge 0$ 
and $\gamma > 0$ in the RNCRN system~(\ref{eq:single_layer_RRE}). 
\begin{enumerate}
\item[\textbf{(a)}] \textbf{Quasi-static approximation}.
Fix a tolerance $\varepsilon > 0$.
Fix also the number of perceptrons $M \ge 1$. Using the backpropagation algorithm~\cite{rumelhart_learning_1986}, 
find the coefficients $\alpha_{i,j}^*, \theta_{j}^*, \omega_{j,i}^*$ 
for $i = 1, 2, \ldots, N$, $j = 1,2,\ldots,M$,
such that (mean-square) distance 
between 
$(f_i(x_1,x_2,\ldots,x_N) - \beta_i)/x_i$ and 
$\sum_{j=1}^M \alpha_{i,j}^* 
\sigma_{\gamma} \left(\sum_{k=1}^N 
\omega_{j,k}^* x_k + \theta_{j}^* \right)$
is within the tolerance for $(x_1,x_2,\ldots,x_N) \in 
 \mathbb{K}_1 \times \mathbb{K}_2 \times \ldots \times\mathbb{K}_N$. If the tolerance $\varepsilon$ is not met,
then repeat step \textbf{(a)} with $M + 1$.

\item[\textbf{(b)}] \textbf{Dynamical approximation}.
Substitute $\alpha_{i,j} = \alpha_{i,j}^*$,
$\theta_{j} = \theta_{j}^*$, $\omega_{j,i} = \omega_{j,i}^*$ into the RNCRN~(\ref{eq:single_layer_RRE}). 
Fix the initial conditions $a_1,a_2,\ldots,a_M \ge 0$, 
$b_1,b_2,\ldots,b_M \ge 0$, and time $T > 0$.
Fix also the speed $0 < \mu \ll 1$ of the perceptrons.
Numerically solve the target system~(\ref{eq:target_ODEs})
and the RNCRN system~(\ref{eq:single_layer_RRE}) 
over the desired interval $[0, T]$.
Time $T > 0$ must be such that 
$\bar{x}_i(t), x_i(t)  \in \mathbb{K}_i$
for all $t \in [0,T]$ for $i = 1, 2,\ldots, N$.
If $\bar{x}_i(t)$ and $x_i(t)$ are sufficiently close
according to a desired criterion for all $i$, 
then terminate the algorithm.
Otherwise, repeat step \textbf{(b)} with a smaller $\mu$.
If no desirable $\mu$ is found, then go back to
step \textbf{(a)} and choose a smaller $\varepsilon$.
\end{enumerate}
\hrule
\caption{{\it Two-step algorithm for training the \emph{RNCRN}.
See~\emph{Appendix}~\ref{sec:pseudocode} for pseudocode.}}
\label{algh:RNCRN}
\end{table}

\subsection{Multi-stability}
\label{sec:multi_stable}
Let us consider the one-variable target ODE
\begin{align}
\frac{\mathrm{d} \bar{x}_1}{\mathrm{d} t}  &= f_1(\bar{x}_1) = \sin(\bar{x}_1), 
\; \; \; \bar{x}_1(0) = a_1 \ge 0.
\label{eq:multistable_ex}
\end{align}
This system has infinitely many equilibria,
which are given by $\bar{x}_1^* = n \pi$
for integer values of $n$. The equilibria
with even $n$ are unstable, while those with 
odd $n$ are stable. 

\textbf{Bi-stability}. 
Let us now apply Algorithm~\ref{algh:RNCRN}
on the target system~(\ref{eq:multistable_ex}),
in order to find an associated bi-stable RNCRN.
In particular, let us choose the target region to be
$\mathbb{K}_1 = [1,12]$, which includes
two stable equilibria, $\pi$ and $3 \pi$, 
and one unstable equilibrium, $2 \pi$. 
We arbitrarily fix the free rate coefficients 
to $\beta_1 = 0$ and $\gamma = 1$.

\emph{\emph{(a)} Quasi-static approximation}.  
Let us apply the first step from Algorithm~\ref{algh:RNCRN}. 
We find that the tolerance $\varepsilon \approx 10^{-3}$ 
can be met with $M = 3$ chemical perceptrons, if the rate coefficients
$\alpha_{1,j}, \theta_{j},\omega_{j,1}$ in the reduced ODE~(\ref{eq:single_layer_RRE_reduced}) are chosen as follows:
\begin{align}
\frac{\mathrm{d} \tilde{x}_1}{\mathrm{d} t} = 
g_1(\tilde{x}_1) = 
& -0.983 \sigma_{1}\left(-1.167\tilde{x}_1 + 7.789 \right) \tilde{x}_1  \nonumber \\ 
& -0.050 \sigma_{1}\left(0.994 \tilde{x}_1 -1.918 \right) \tilde{x}_1 \nonumber \\
& + 2.398 \sigma_{1}\left(-0.730 \tilde{x}_1 +  3.574\right) \tilde{x}_1.
\label{eq:single_layer_RRE_reduced_sin_bi_stable}
\end{align}
In Figure~\ref{fig:bi_stable_sin_dynamics_and_traj}(a), 
we display the vector fields of the target~(\ref{eq:multistable_ex}) 
and reduced system~(\ref{eq:single_layer_RRE_reduced_sin_bi_stable}).
One can notice an overall good match within the desired set
$\mathbb{K}_1 = [1,12]$, shown as the unshaded region.
As expected, the approximation is poor outside of $\mathbb{K}_1$;
furthermore, for the given tolerance, the accuracy 
is also reduced near the right end-point of the target set.

\begin{figure}[ht]
    \centering
    \includegraphics[width=0.69\columnwidth]{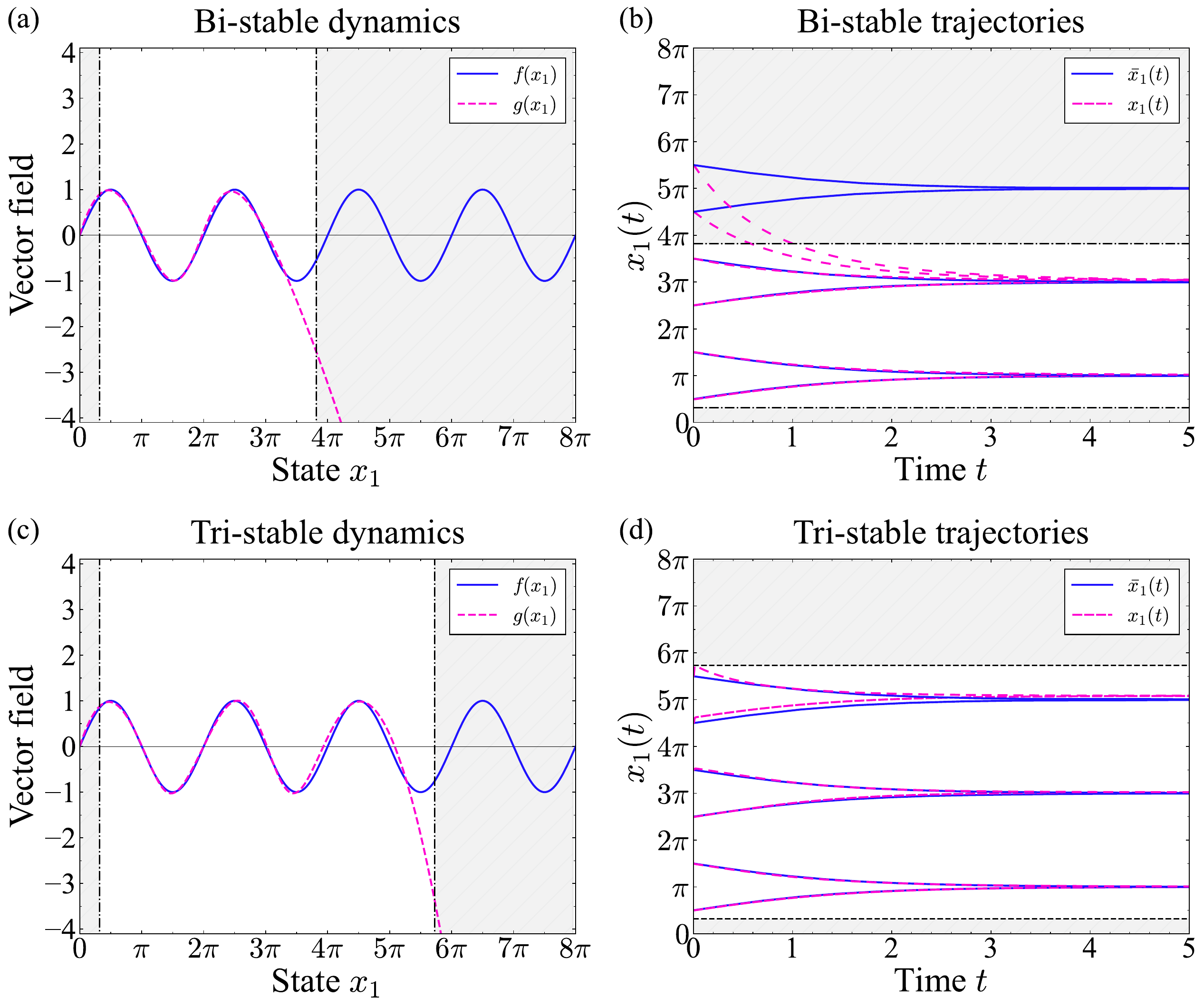}
    \caption{RNCRN approximations of the multi-stable target system~(\ref{eq:multistable_ex}). (a) The vector field of the target system~(\ref{eq:multistable_ex}) and reduced bi-stable system~(\ref{eq:single_layer_RRE_reduced_sin_bi_stable}) when $\mathbb{K}_1 = [1,12]$. (b) Solutions $\bar{x}_1(t)$ of the target system~(\ref{eq:multistable_ex}), and $x_1(t)$ of the full system~(\ref{eq:single_layer_RRE_full_sin_bi_stable}) with $\mu = 0.01$. Analogous plots are shown in panels (c) and (d) 
    for the tri-stable RNCRN over $\mathbb{K}_1 = [1,18]$ 
    presented in Appendix~\ref{app:app_tri_stable},
    with reduced and full ODEs given respectively by~(\ref{eq:single_layer_RRE_reduced_sin_tri_stable})
    and (\ref{eq:single_layer_RRE_full_sin_tri_stable}), 
    with coefficients~(\ref{eq:tri_stable_coeff})
    and $\mu = 0.01$. In panels (b) and (d), 
    the initial concentrations of all chemical perceptrons are set to zero.
    }
    \label{fig:bi_stable_sin_dynamics_and_traj}
\end{figure}

\emph{\emph{(b)} Dynamical approximation}. Let us apply the second step from Algorithm~\ref{algh:RNCRN}. Using the coefficients from~(\ref{eq:single_layer_RRE_reduced_sin_bi_stable}), 
we now form the full ODEs~(\ref{eq:single_layer_RRE}):
\begin{align}
\frac{\mathrm{d} x_1}{\mathrm{d}t} & = -0.983 x_1 y_1 - 0.050 x_1 y_2 + 2.398 x_1 y_3,
 && x_1(0) = a_1 \in \mathbb{K}_1, \nonumber\\
\frac{\mathrm{d} y_1}{\mathrm{d}t} & = \frac{1}{\mu} + \frac{7.789}{{\mu}} y_1 - 
\frac{1.167}{\mu} x_1 y_1 - \frac{1}{\mu} y_1^2, 
&& y_1(0) = b_1 \geq 0, \nonumber \\
\frac{\mathrm{d} y_2}{\mathrm{d}t} & = 
\frac{1}{\mu} - \frac{1.918}{\mu} y_2 + \frac{0.994}{\mu} x_1 y_2 - \frac{1}{\mu} y_2^2, && y_2(0) = b_2  \geq 0, \nonumber \\
\frac{\mathrm{d} y_3}{\mathrm{d}t} & = \frac{1}{\mu} + \frac{3.574}{\mu} y_3 - \frac{0.730}{\mu} x_1 y_3 - \frac{1}{\mu} y_3^2, 
&& y_3(0) = b_3  \geq 0.
\label{eq:single_layer_RRE_full_sin_bi_stable}
\end{align}
We fix the initial conditions arbitrarily to 
$b_1 = b_2 = b_3 = 0$, the desired final-time to $T = 5$
and the perceptron speed to $\mu=0.01$. 
We then numerically integrate~(\ref{eq:multistable_ex}) and (\ref{eq:single_layer_RRE_full_sin_bi_stable}) over 
$t \in [0,5]$ for a fixed initial concentration $a_1  \in \mathbb{K}_1$
of the executive species, 
and plot the solutions $\bar{x}(t)$ and $x(t)$;
the same computations are then repeated for various values of $a_1$,
which we display in 
Figure~\ref{fig:bi_stable_sin_dynamics_and_traj}(b).
One can notice that the RNCRN underlying~(\ref{eq:single_layer_RRE_full_sin_bi_stable}) with $\mu=0.01$ accurately approximates the solutions 
of the target system; in particular, we observe bi-stability. 

\textbf{Tri-stability}. 
One can similarly apply Algorithm~\ref{algh:RNCRN}
to~(\ref{eq:multistable_ex}) to obtain tri-stable RNCRNs.
In particular, let us consider a larger
set $\mathbb{K}_1 = [1,18]$, which includes
three stable equilibria, $\pi$, $3 \pi$, and $5 \pi$, 
and two unstable equilibria, $2 \pi$ and $4 \pi$.
Applying Algorithm~\ref{algh:RNCRN}, we find 
an RNCRN with $M = 4$ chemical perceptrons,
which with $\mu = 0.01$ displays the desired tri-stability,
as displayed in Figure~\ref{fig:bi_stable_sin_dynamics_and_traj}(c)--(d);
see Appendix~\ref{app:app_tri_stable} for more details.
In a similar manner, Algorithm~\ref{algh:RNCRN}
can be used to achieve RNCRNs with arbitrary number of stable equilibria; see Appendix~\ref{app:app_tri_stable}.

\textbf{DNA-strand-displacement implementation}.
As outlined in Section~\ref{sec:crns_background}, 
the RNCRN consists of abstract reactions, 
which must be appropriately enlarged 
for experimental implementations.
As a demonstration, let us implement the RNCRN
corresponding to the RREs~(\ref{eq:single_layer_RRE_full_sin_bi_stable})
via an experimentally feasible DNA-strand-displacement system~\cite{soloveichik_dna_2010}.
In this framework, e.g. the abstract reaction 
$X_1 + Y_2 \xrightarrow[]{1.988} X_1 + 2 Y_2$,
induced by the term $(0.994/\mu) x_1 y_2$ with $\mu = 0.5$ 
from~(\ref{eq:single_layer_RRE_full_sin_bi_stable}), 
is approximated with a system of $4$ reactions
involving $7$ species, given by 
\begin{align} 
X_1 + L_{1} & \xrightleftharpoons[\textrm{fast}]{1.988} H_{1} + B_{1}, 
\; \; \; \; 
Y_2 + H_{1} \xrightarrow{\textrm{fast}} O_{1},  
\; \; \; \; 
O_{1} + T_{1} \xrightarrow{\textrm{fast}} X_1 + 2 Y_2. 
\label{eq:DNA_example}
\end{align}
Here, we denote the two irreversible reactions 
$X_1 + L_{1} \xrightarrow{1.988} H_{1} + B_{1}$
and 
$H_{1} + B_{1}\xrightarrow{\textrm{fast}} X_1 + L_{1}$
as a single reversible one; 
furthermore, some of the rate coefficients
are described qualitatively as sufficiently ``fast'',
see Appendix~\ref{sec:implementation} for more details.
Species $X_1$, $Y_1$, $B_{1}$ and $O_1$ can be interpreted  
as suitable single-stranded DNA molecules, 
while $L_{1}$, $H_{1}$, $T_{1}$ 
as double-stranded DNA complexes;
see Figure~\ref{fig:bi_stable_dna_implementation}(a), 
where we display the DNA-based approximation~(\ref{eq:DNA_example}) schematically. 

In Appendix~\ref{sec:implementation}, we approximate
every reaction from the RNCRN induced by RREs~(\ref{eq:single_layer_RRE_full_sin_bi_stable})
 with a system of DNA-based reactions, 
such as~(\ref{eq:DNA_example});
furthermore, we fix the perceptron speed to $\mu = 0.5$.
For the resulting network, which we call 
the DNA-RNCRN, we assume that 
time and concentrations have units;
furthermore, we rescale these quantities so that 
the concentrations of the DNA-species
and rate coefficients fit within
the physically admissible ranges~\cite{soloveichik_dna_2010}.
By design, the DNA-RNCRN relies on suitable 
DNA species (e.g. $L_1, B_1, T_1$ from~(\ref{eq:DNA_example}))
being present at sufficiently high concentrations.
Being depleted over time, these fuel species
must in general be replenished, e.g. by coupling the DNA-RNCRN to a chemostat~\cite{soloveichik_dna_2010}. 
In this paper, we do not model the replenishment
process explicitly; instead, for simplicity, we 
assume that the fuel species are held constant.
The resulting approximating \emph{constant-fuel} RNCRN is given 
by~(\ref{eq:dna_simple_rncrn})
in Appendix~\ref{sec:implementation}. 
In Figure~\ref{fig:bi_stable_dna_implementation}(b), we display the concentration of the executive species from the
constant-fuel DNA-RNCRN, demonstrating that bistability 
is accurately replicated within experimentally feasible range. 
See Appendix~\ref{sec:implementation} for more details, 
including a simulation of the DNA-RNCRN without
the constant-fuel assumption.

\begin{figure}[H]
    \centering
    \includegraphics[align=t, width=0.66\columnwidth, height=4.0cm, trim={0cm 0cm 1.0cm, 0cm},clip]{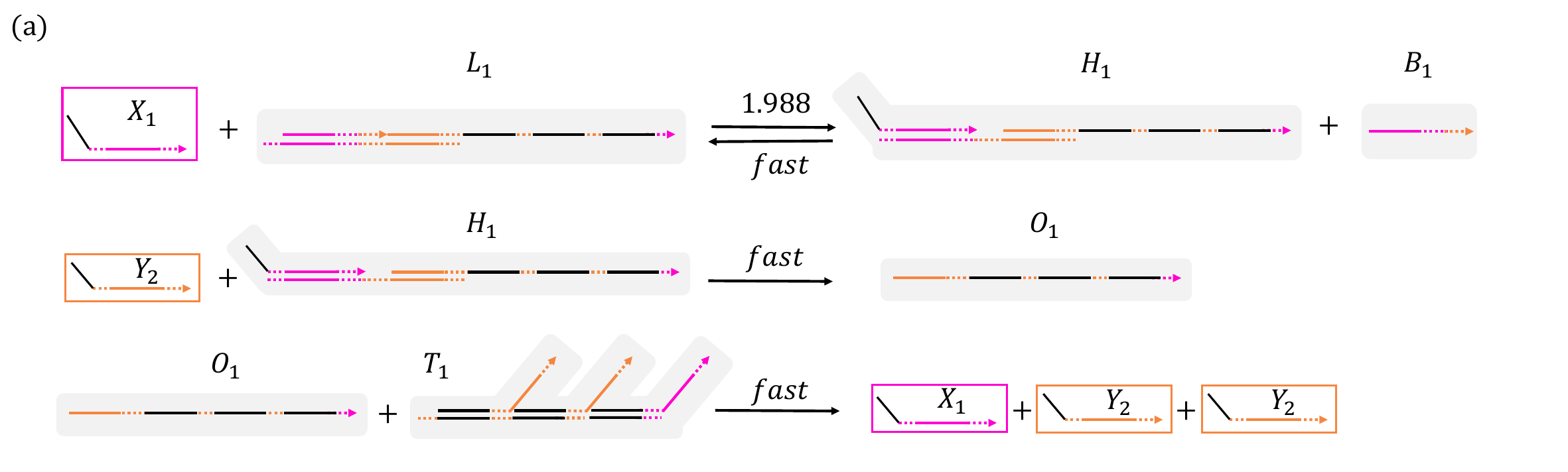}
    \includegraphics[align=t, width=0.33\columnwidth, height=4.5cm ]{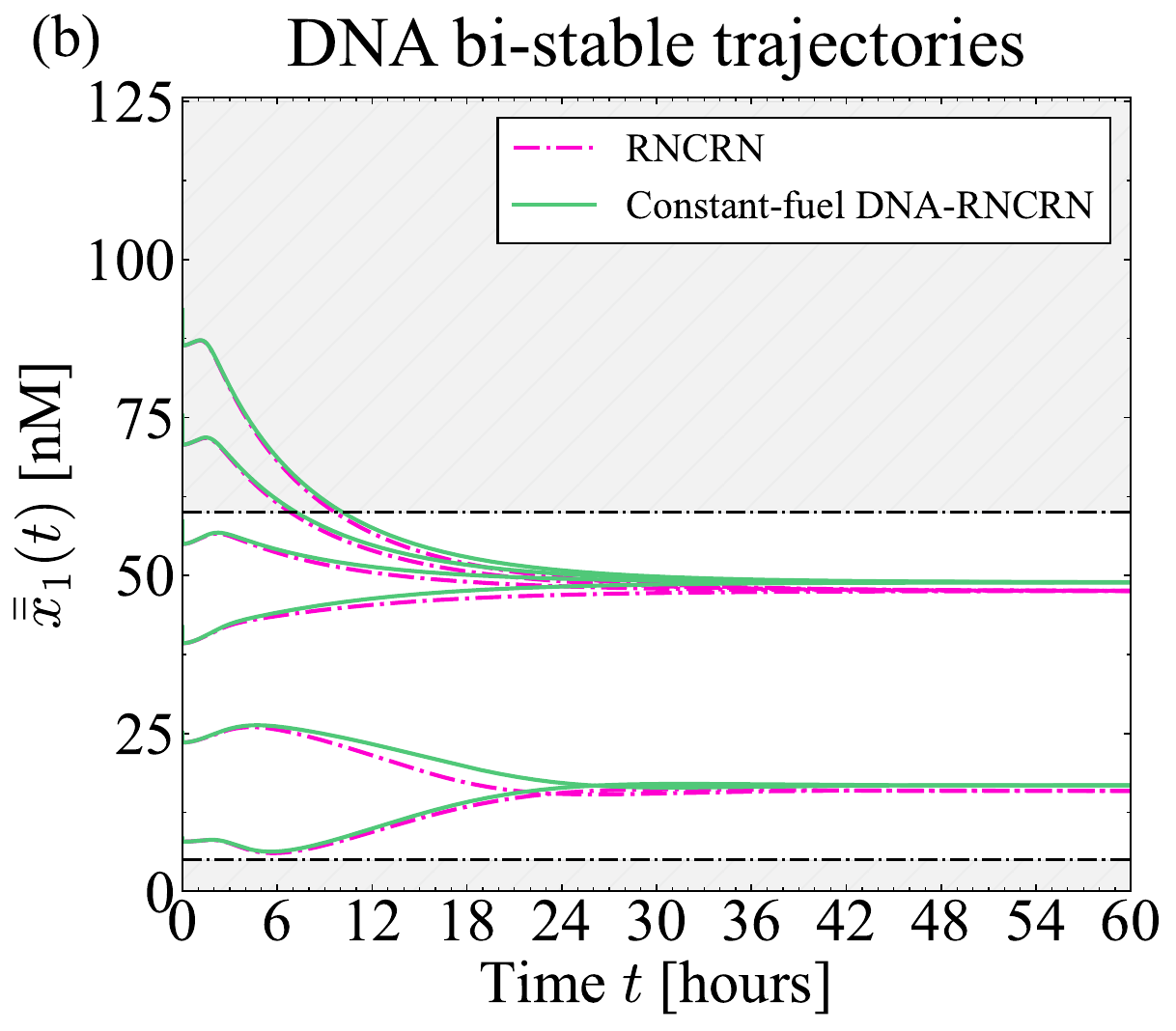}
    \caption{DNA-strand-displacement implementation of the RNCRN
    with RREs~(\ref{eq:single_layer_RRE_full_sin_bi_stable})
    with $\mu = 0.5$. 
    (a) Domain-level representation of the DNA system~(\ref{eq:DNA_example}), which approximates the bi-molecular reaction $X_1 + Y_2 \xrightarrow[]{1.988} X_1 + 2 Y_2$. 
    (b) Shown in purple is the concentrations $\bar{\bar{x}}_1(t)$
    of the executive species from a rescaled version of the RREs~(\ref{eq:single_layer_RRE_full_sin_bi_stable}), 
    given by~(\ref{eq:single_layer_RRE_full_sin_bi_stable_rescaled}) in Appendix~\ref{sec:implementation}. 
    Also shown in green is the corresponding concentration of the executive species from the constant-fuel RNCRN given by~(\ref{eq:dna_simple_rncrn}) in Appendix~\ref{sec:implementation}. The concentrations are measured in nanomoles (nM), while time in hours.
    }
    \label{fig:bi_stable_dna_implementation}
\end{figure}

\subsection{Oscillations}
\label{sec:oscillations}
Let us consider the two-variable target ODE system  
\begin{align}
\frac{\mathrm{d} \bar{x}_1}{\mathrm{d} t} &= f_1(\bar{x}_1, \bar{x}_2) = 6 + 4 J_0\left(\frac{3}{2}\bar{x}_1\right) - \bar{x}_2,
&& \bar{x}_1(0) = a_1 \geq 0, \nonumber\\
    \frac{d \bar{x}_2}{d t} &= f_2(\bar{x}_1, \bar{x}_2) = \bar{x}_1 -4,
    && \bar{x}_2(0) = a_2 \geq 0,
    \label{eq:bessel_osc}
\end{align}
were $J_0(\bar{x}_1)$ is the Bessel function of the first kind. Numerical simulations suggest that~(\ref{eq:bessel_osc})
has an isolated oscillatory solution,
which we display as the blue curve in the $(\bar{x}_1,\bar{x}_2)$-space
in Figure~\ref{fig:bessel_osc_traj_cherr_net}(a);
also shown as grey arrows is the vector field, 
and as the unshaded box we display the 
desired region of interest $\mathbb{K}_1
\times \mathbb{K}_2 = [1,8] \times [2, 10]$.
In Figure~\ref{fig:bessel_osc_traj_cherr_net}(c),
we show as solid blue curves this oscillatory solution
in the $(t,\bar{x}_1)$- and $(t,\bar{x}_2)$-space.

Using Algorithm~\ref{algh:RNCRN}, we find an RNCRN
with $M = 6$ chemical perceptrons, whose dynamics
with $\mu = 0.01$ qualitatively matches
that of the target system~(\ref{eq:bessel_osc})
within $\mathbb{K}_1\times \mathbb{K}_2$;
see Appendix~\ref{sec:app_bessel_osc} for details.
We display the solution of this RNCRN 
as the purple curve in 
Figure~\ref{fig:bessel_osc_traj_cherr_net}(b)--(c). 

\begin{figure}[H]
    \centering
    \includegraphics[width=\columnwidth]{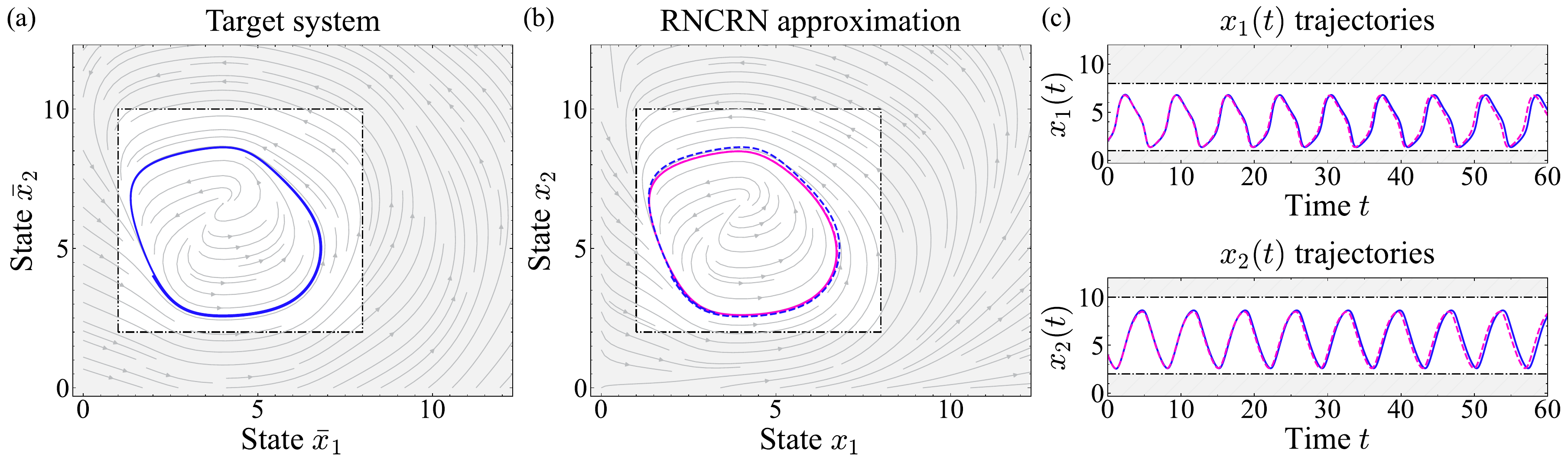}
\caption{RNCRN approximation of the oscillatory target system~(\ref{eq:bessel_osc}). (a) The vector field of~(\ref{eq:bessel_osc}) 
is shown as grey arrows, the target region $\mathbb{K}_1 \times \mathbb{K}_2 = [1,8] \times [2, 10]$ as the unshaded box, and the solution of~(\ref{eq:bessel_osc}) with $\bar{x}_1(0)=2$ and $\bar{x}_2(0)=4$ in blue. 
(b) Analogous plot is shown for the RNCRN from Appendix~\ref{sec:app_bessel_osc}, whose reduced and full ODEs are given respectively by~(\ref{eq:single_layer_RRE_reduced_oscill}) and~(\ref{eq:full_rncrn_oscil}), 
with coefficients~(\ref{eq:coefficents_oscill}) and $\mu = 0.01$.
In particular, displayed as grey arrows is the vector field
of the reduced ODEs~(\ref{eq:single_layer_RRE_reduced_oscill}), 
together with an oscillatory solution of the full ODEs~(\ref{eq:full_rncrn_oscil}) shown in purple; 
for comparison, we also display as dashed blue curve 
the oscillatory solution of~(\ref{eq:bessel_osc}).
(c) Solutions $\bar{x}_1(t)$ and $x_1(t)$, 
and $\bar{x}_2(t)$ and $x_2(t)$. 
The initial concentrations of all perceptrons 
are fixed to zero. 
}
\label{fig:bessel_osc_traj_cherr_net}
\end{figure}

\textbf{Robustness}. When implementing 
RNCRNs, the underlying rate coefficients 
and initial conditions cannot be experimentally fine-tuned 
with perfect accuracy. Therefore, it is of great importance
to study how the dynamics of RNCRNs behaves 
under the parameter perturbations. 
A detailed analysis of this problem is beyond the scope of this paper.
As a step forward, we numerically investigate
sensitivity of the oscillatory solution of the 
RNCRN from Appendix~\ref{sec:app_bessel_osc},
which approximates the target system~(\ref{eq:bessel_osc});
see Appendix~\ref{sec:robustness} for details.
In particular, we perturb all the rate coefficients
in the RNCRN by up to $50\%$, and compute the fraction 
of such perturbed RNCRNs that qualitatively retain 
the periodic solution as a function of the perturbation noise. 
The results for the RNCRN with $\mu = 0.1$
are presented in Figure~\ref{fig:bessel_osc_traj_cherr_net_robustness}(a)
in purple, with example trajectories in Figure~\ref{fig:bessel_osc_traj_cherr_net_robustness}(b); 
we observe that e.g. around $35\%$ of systems 
still oscillate when the
coefficients are perturbed by $10\%$.
Also shown in Figure~\ref{fig:bessel_osc_traj_cherr_net_robustness}(a) 
in green are the results for the underlying quasi-static approximation,
formally obtained by setting $\mu = 0$.
One can notice that this approximation shows a similar
degree of robustness, suggesting that the perceptron speed
 is not a major source of sensitivity.
This observation is confirmed in Figure~\ref{fig:bessel_osc_traj_cherr_net_robustness}(c)--(d), 
which shows that the RNCRN is insensitive to 
a broad range of perturbations in $\mu$. 
Robustness over a broad range of perturbations 
can also be observed with respect to the initial conditions, 
see Figure~\ref{fig:bessel_osc_traj_cherr_net_robustness}(e)--(f).

\emph{Noisy training}. 
Figure~\ref{fig:bessel_osc_traj_cherr_net_robustness}
demonstrates that the RNCRN from Appendix~\ref{sec:app_bessel_osc}, 
while not being pathologically
sensitive to the perturbations in the rate coefficients, 
is nevertheless much more vulnerable to such perturbations
than those in perceptron speed and initial conditions.
Let us now demonstrate that, by slightly modifying Algorithm~\ref{algh:RNCRN}, the rate-coefficient 
robustness can be significantly improved. In particular, step (a)
of Algorithm~\ref{algh:RNCRN}
approximates the target vector
field from~(\ref{eq:bessel_osc}) 
with the deterministic reduced vector field from~(\ref{eq:single_layer_RRE_reduced}).
In Appendix~\ref{app:noisy_training}, we modify this step,
by introducing suitable stochastic perturbations 
into the reduced vector field. 
We apply this noisy version of Algorithm~\ref{algh:RNCRN}
to design a new RNCRN, presented in Appendix~\ref{app:noisy_training},
which also contains $M = 6$ chemical perceptrons
and approximates~(\ref{eq:bessel_osc}).
This noise-trained RNCRN displays a significantly 
better robustness profile, 
as demonstrated in yellow in Figure~\ref{fig:bessel_osc_traj_cherr_net_robustness}(a);
in particular, now almost all of the systems still oscillate
when the rate coefficients are perturbed by $10\%$, 
and more than $60\%$ oscillate when the coefficients
are perturbed by as much as $50\%$.

\begin{figure}[H]
\centering
\includegraphics[width=\columnwidth]{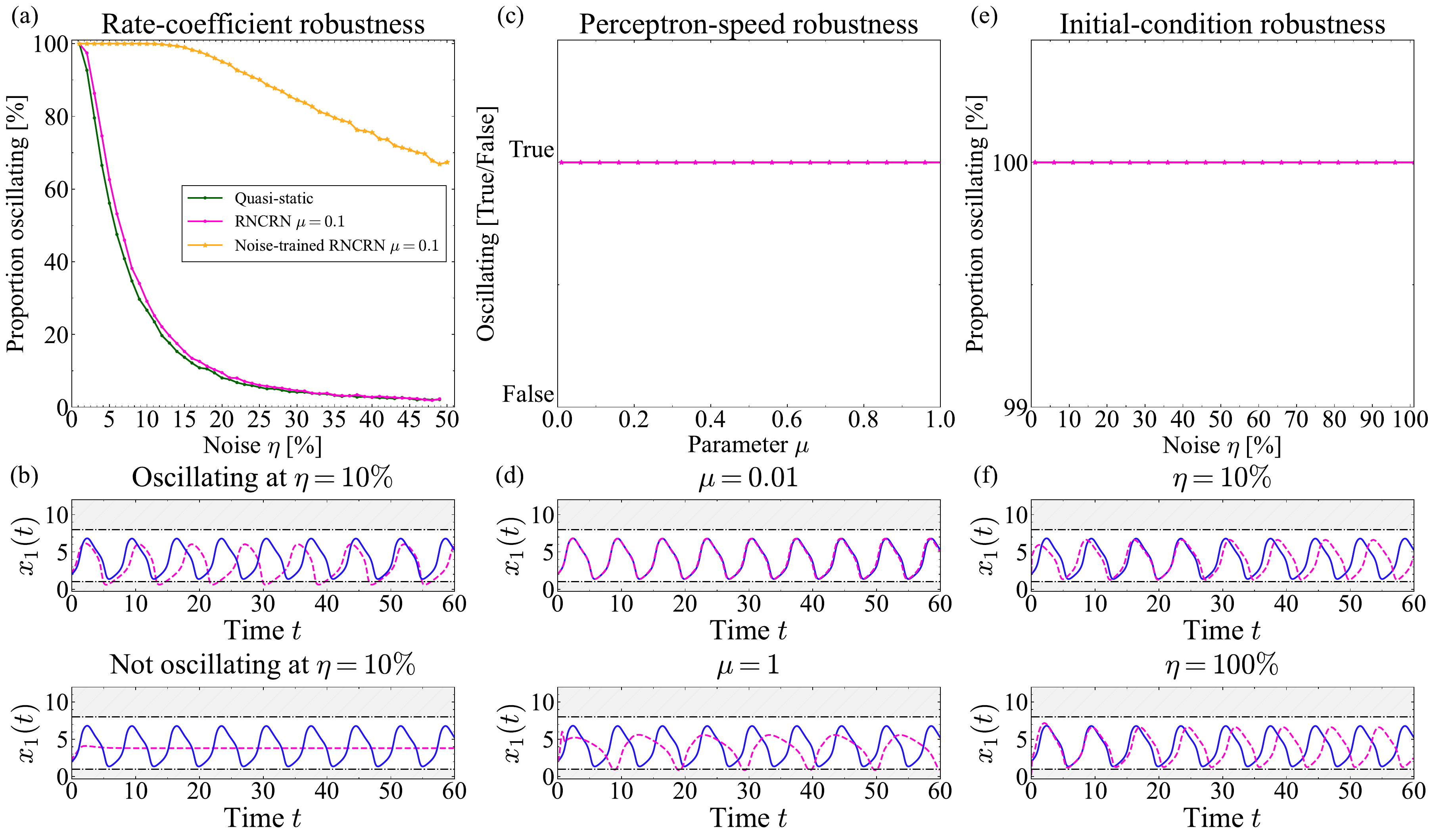}
\caption{Robustness of RNCRNs approximating
oscillatory target system~(\ref{eq:bessel_osc}). 
(a) Proportion of RNCRNs with perturbed
rate coefficients that retain oscillations as a function of 
the perturbation noise. This proportion 
is shown in green for the quasi-static approximation~(\ref{eq:single_layer_RRE_reduced_oscill}), 
in purple for the corresponding RNCRN~(\ref{eq:full_rncrn_oscil})
with $\mu=0.1$, and in yellow for the noise-trained RNCRN~(\ref{eq:full_rncrn_oscil_noise}) with $\mu=0.1$; 
see Appendix~\ref{sec:rate_robustness} and~\ref{app:noisy_training} 
for more details.
Example trajectories of the RNCRN~(\ref{eq:full_rncrn_oscil})
which retain and do not retain oscillations
are shown in purple in panel (b), together with the corresponding
solutions of~(\ref{eq:bessel_osc}) in blue.
(c) Preservation of oscillations in 
the RNCRN~(\ref{eq:full_rncrn_oscil}) as parameter $\mu$ is varied,
with example trajectories shown in panel (d);  
see Appendix~\ref{app:perceptron_speed_robustness} for more details.
(e) Proportion of the RNCRNs~(\ref{eq:full_rncrn_oscil}) 
with perturbed initial conditions that retain oscillations,
with example trajectories shown in panel (f);
see Appendix~\ref{app:init_conc_robustness} for more details.
Each point in panels (a), (c) and (e) was obtained 
by considering $10^4$ independently perturbed RNCRNs.}
\label{fig:bessel_osc_traj_cherr_net_robustness}
\end{figure}

\subsection{Chaos}
As our final example, we consider the three-variable target ODE system
\begin{align}
    \frac{\mathrm{d}\bar{x}_1}{\mathrm{d}t} &=  f_1(\bar{x}_1, \bar{x}_2, \bar{x}_3) =2.5\bar{x}_1(1-1.5\bar{x}_1) -\frac{4 \bar{x}_1\bar{x}_2}{1+3\bar{x}_1}, && \bar{x}_1(0) = 0.25,\nonumber \\
    \frac{\mathrm{d}\bar{x}_2}{\mathrm{d}t} &=  f_2(\bar{x}_1, \bar{x}_2, \bar{x}_3) = -0.4\bar{x}_2+\frac{4 \bar{x}_1\bar{x}_2}{1+3 \bar{x}_1}-\frac{4 \bar{x}_2\bar{x}_3}{1+3 \bar{x}_2},  && \bar{x}_2(0) = 0.25, \nonumber\\
    \frac{\mathrm{d}\bar{x}_3}{\mathrm{d}t} &= f_3(\bar{x}_1, \bar{x}_2, \bar{x}_3) = -0.6\bar{x}_3+ \frac{4 \bar{x}_2\bar{x}_3}{1+3 \bar{x}_2},  && \bar{x}_3(0) = 0.25.
    \label{eq:hasting_powell_ode}
\end{align}
ODEs~(\ref{eq:hasting_powell_ode}) are known as the Hasting-Powell system~\cite{stone_chaotic_2007},
and have been reported to exhibit chaotic behaviour.
We present a portion of the state-space of~(\ref{eq:hasting_powell_ode})
in Figure~\ref{fig:hasting_powell_dynamics_and_traj}(a).
 
Using Algorithm~\ref{algh:RNCRN}, we find 
an RNCRN with $M = 5$ chemical perceptrons to approximate
the dynamics of the Hasting-Powell system~(\ref{eq:hasting_powell_ode})
over the compact set $x_1 \in \mathbb{K}_1$, 
$x_2 \in \mathbb{K}_2$ and $x_2 \in \mathbb{K}_3$
with $\mathbb{K}_1 = \mathbb{K}_2 = \mathbb{K}_3 = [0.01,1]$.
The parameters and the RREs can be found in Appendix~\ref{sec:hasting_powell_app}. 
The state-space for the executive species from the RNCRN is presented 
in Figure~\ref{fig:hasting_powell_dynamics_and_traj}(b)
when $\mu = 0.1$, suggesting presence of a chaotic attractor,
in qualitative agreement with 
Figure~\ref{fig:hasting_powell_dynamics_and_traj}(a).
In Figure~\ref{fig:hasting_powell_dynamics_and_traj}(c), 
we display the underlying time-trajectories. 
Let us note that alignment of the trajectories is
not expected for long intervals of time 
due to the chaotic nature of the target system.
\begin{figure}[ht]
\centering
\includegraphics[width=\columnwidth]{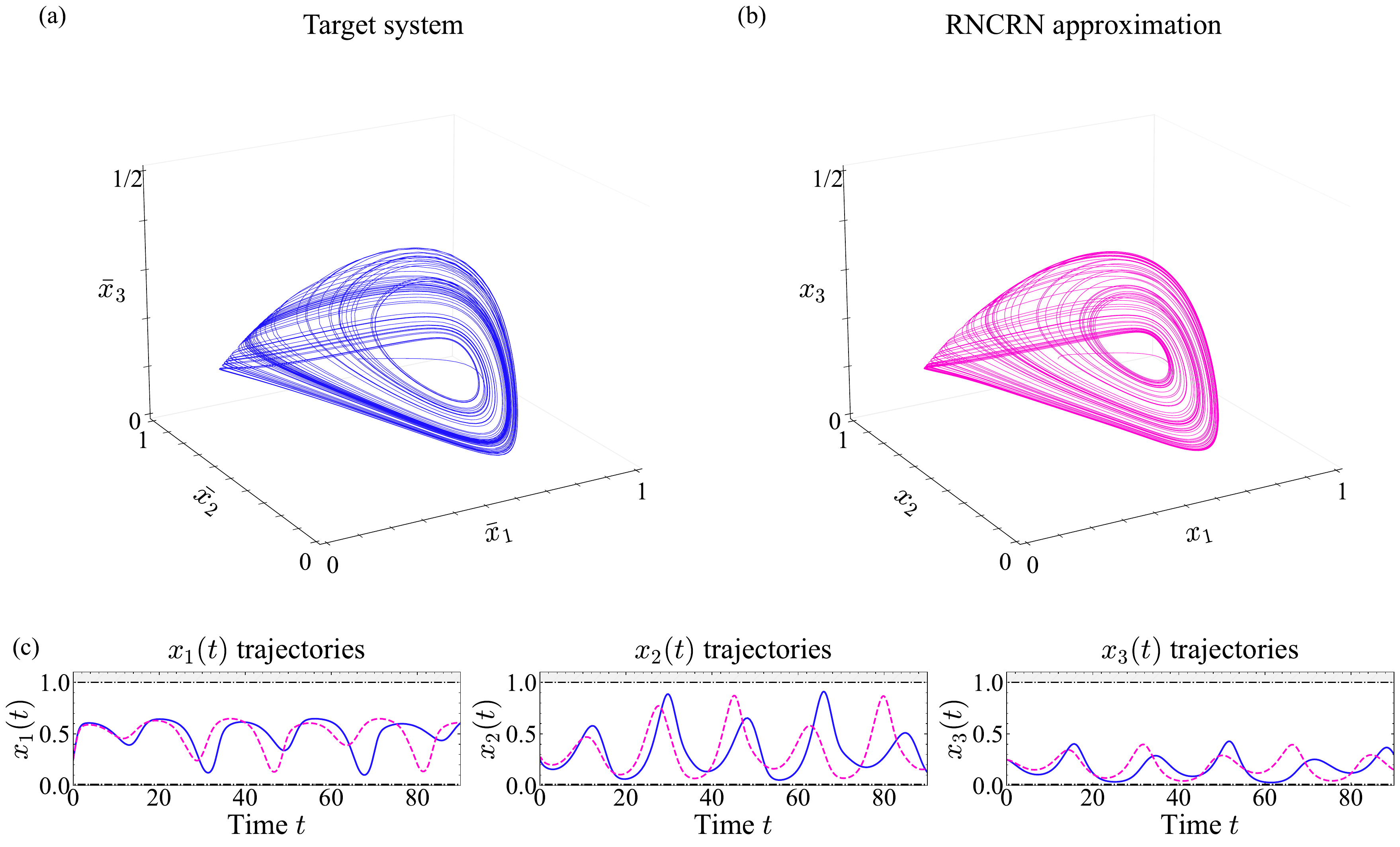}
\caption{RNCRN approximation of the target system~(\ref{eq:hasting_powell_ode}) with a chaotic attractor. 
(a) Solution of~(\ref{eq:hasting_powell_ode}) in the state-space. 
(b) Solution of the RNCRN RREs~(\ref{eq:full_rncrn_chaos}) in the
state-space, with coefficients~(\ref{eq:coeff_chaos}), $\mu = 0.1$, 
initial conditions $x_1(0) = x_2(0) = x_3(0) = 0.25$, 
and zero initial conditions for all perceptrons. 
(c) Target and executive solutions $\bar{x}_i(t)$ 
and $x_i(t)$ for $i = 1,2,3$.}
\label{fig:hasting_powell_dynamics_and_traj}
\end{figure}

\section{Discussion} \label{sec:con_disc}
In this paper, we have introduced a class of chemical reaction networks 
(CRNs) that can be trained to approximate the dynamics of any well-behaved system of ordinary-differential equations (ODEs). These recurrent neural chemical reaction networks (RNCRNs) operate by continually actuating the concentrations of the executive species such that they adhere to the target vector field which has been encoded by a 
faster system of chemical perceptrons.
In Appendix~\ref{app:single_layer_RNCRN}, we have proven
the approximation abilities of RNCRNs. Based on this result,  we have put forward Algorithm~\ref{algh:RNCRN} for training RNCRNs, which relies
on the backpropagation algorithm from the theory
of artificial neural networks (ANNs). Due to the nature of backpropagation, Algorithm~\ref{algh:RNCRN} is not guaranteed to find an optimal solution; nevertheless, it proved effective for the target systems presented in this paper. In particular, in Section~\ref{sec:examples}, we showcased examples of RNCRNs that replicate sophisticated dynamical behaviors with no more than six chemical perceptrons, governed by reactions which fire 
at most two orders of magnitude faster than
those governing the executive species. 

For target systems of ODEs with purely polynomial right-hand sides, 
there exist a number of other methods, which we call \emph{chemical maps}, to design dynamically-equivalent CRNs~\cite{plesa_quasichemical_2025,plesa_chemical_2016,samardzija_nonlinear_1989,poland_cooperative_1993,  holcman_test_2017,plesa_integral_2021, hangos_mass_2011,kowalski_universal_1993}.
CRNs designed using these chemical maps have analytically
determined rate coefficients, i.e. they do not require
numerical methods, such as the backpropagation algorithm,
for this purpose; furthermore, these CRNs
can involve less auxiliary species than RNCRNs.
However, an advantage of RNCRNs
is their ANN-based ability to replicate the dynamics of,
not only ODEs with polynomial,
but also non-polynomial right-hand sides, 
the latter of which can display a richer set of 
dynamical behaviors compared to the former.
In Section~\ref{sec:examples}, we have exploited this fact
by considering example ODEs that contain 
a mixture of polynomial and non-polynomial terms.
Let us also note that the structure of CRNs obtained via chemical maps
depends on the given target ODE system.
On the other hand, RNCRNs have a fixed structure, 
independent of the complexity of the target system, 
only relying on the reactions from Table~\ref{table:example_CRN}; 
this modularity may be beneficial for experimentally building RNCRNs. 

For target systems of ODEs with non-polynomial right-hand sides,
there also exists an alternative scheme.
In particular, under this scheme, one first maps the non-polynomial 
ODEs to polynomial ones~\cite{kerner_universal_1981}, 
which is then followed by an application of a desired chemical map.
Similar to RNCRNs, this scheme also introduces
auxiliary variables to achieve the desired dynamics. 
However, in stark contrast to RNCRNs, 
owing to the map from~\cite{kerner_universal_1981}, 
this alternative scheme displays a particular fragility:
the initial conditions for the auxiliary species are constrained
to particular values, and some deviations may cause
some species concentrations to blow-up.
In Appendix~\ref{sec:compare}, we
demonstrate such chemically hazardous events
when the scheme is applied to the 
target system~(\ref{eq:multistable_ex}).
This fragility not only requires high-precision calibration
of the initial conditions, but also effectively requires the preparation of an entirely different experiment for different executive species initial conditions. In contrast, no such adjustment of the auxiliary species 
is necessary for RNCRNs.

While certain properties of RNCRNs are experimentally desirable,
such as modularity and robustness in initial concentrations
of the auxiliary species
(which form the bulk of initial conditions), 
there are several challenges present in chemistry that are not present 
in traditional electronic-based ANNs; in what follows, 
we focus on DNA-strand-displacement technologies~\cite{soloveichik_dna_2010}.
Firstly, adding a perceptron or a new connection to an ANN is trivial compared to adding a chemical perceptron or a reaction to an RNCRN, which requires 
engineering a chemical species with specific reactivity. 
Secondly, compared to electronic circuits,
chemical systems are much more prone to unintended interactions, 
so that the modularity of each chemical perceptron might break 
down as the network scales in size~\cite{qian_scaling_2011}. 
Finally, the weight parameters and inputs of an ANN can be fine-tuned to a high degree of precision. However, despite techniques to 
modulate the effective rates of chemical reactions, the calibration is significantly less precise than in electronic circuits~\cite{erik_yu_toehold_09, haley_design_2020},
and the same is true for the initial conditions~\cite{srinivas_enzyme-free_2017}. \Alex{In this context, it is important to study 
how the dynamics of RNCRNs respond to  perturbations
in both the rate coefficients and all the initial conditions. 
In Section~\ref{sec:oscillations}
and Appendix~\ref{sec:robustness}, 
we have performed a preliminary study of this problem, 
showing favorable results; more detailed studies
will form a part of future work.
Let us also note that, in the light of the
experimental challenges, we have focused in 
this paper on the single-layer (shallow) RNCRNs.
In Appendix~\ref{app:single_layer_RNCRN}, we have 
also presented multi-layer (deep) RNCRNs;
in future work, we will study whether such deep RNCRNs have 
sufficiently advantageous approximation abilities 
to justify the additional experimental complexity.}

Despite the various experimental challenges, 
DNA-strand-displacement technologies have been 
used to build intricate chemical systems of great complexity.
In particular, under this framework, arbitrary bi-molecular CRNs, 
whose rate coefficients span up to 
six orders of magnitude~\cite{erik_yu_toehold_09},
can be implemented, and this approach has been used to build 
chemical systems with more than 100 species~\cite{cherry_scaling_2018, xiong_molecular_2022}.
In Section~\ref{sec:multi_stable}
and Appendix~\ref{sec:implementation}, 
we have implemented an RNCRN via DNA-strand-displacement 
under experimentally feasible conditions.
In this context, a natural extension is 
the addition of enzymes, which can enhance 
the capabilities of the DNA circuits.
For example, RNA polymerase has been used to 
continually produce strand displacement 
components~\cite{bae_situ_2021, schaffter_cotranscriptionally_2022}, 
while strand-displacing polymerase~\cite{thachuk_implementing_2019} and 
the trio of polymerase, endonuclease and nickase (PEN)~\cite{okumura_nonlinear_2022,fujii_predatorprey_2013} can
also be used to realize a range of chemical reactions.
Such enzyme-based approaches have the potential 
to enable more universal, easily supplied fuel 
(in the form of nucleotide triphosphates) 
to maintain the concentrations of suitable auxiliary 
components for longer intervals of time,
as well as to speed-up some of the underlying reactions.

We have used DNA strand-displacement
to demonstrate experimental implementability of RNCRNs.
More broadly, RNCRNs can be implemented
with different physical media, provided the existence
of a compilation framework from CRNs to such media~\cite{chen_programmable_2013, soloveichik_dna_2010, thachuk_implementing_2019}. Therefore, the results
presented in this paper have the potential to 
generalize to other nanotechnological and synthetic biological contexts.
In particular, given that RNCRNs contain slow and fast species, 
it is natural to consider architectures that operate over multiple time-scales. Important examples are transcription-factor networks.
These networks have been shown to leverage the innate difference 
in the speed between the slow protein production and degradation, 
and the fast protein-protein interactions~\cite{zhu_synthetic_2022};
furthermore, the transcription-factor binding sites 
display dose thresholding reminiscent of a perceptron~\cite{spivakov_spurious_2014}.
In this context, one can envisage the total concentration
of each protein as the (slow) executive species concentration, 
and the concentrations of particular conformational states of the proteins as the (fast) chemical perceptron concentrations.
Since the ODE models of transcription-factor networks
have on the right-hand side rational functions, and not polynomials, 
the underlying activation function 
would be different from the one considered in this paper.
Nevertheless, the results we present in Appendix~\ref{app:single_layer_RNCRN}
readily generalize to any non-polynomial activation function, 
allowing for construction of RNCRN-like networks
based on transcription factors.
Beyond chemistry and genetic circuits, 
RNCRNs may also find applications 
in immunology~\cite{farmer_immune_1986} 
and population dynamics~\cite{hirafuji_lotka-volterra_1999}.
In particular, let us note that the ODEs describing microbial cellular communities~\cite{stein_ecological_2013} are
reminiscent of those governing chemical 
perceptrons presented in this paper.

\section{Acknowledgements}
\label{app:acknowledge}
\noindent
{\small {\bf Author Contributions}:
AD, BQ, TEO, and TP conceptualization the study; 
AD and TP preformed the mathematical analyses; 
AD preformed the numerical simulations and produced the figures; 
AD and TP wrote the paper;
BQ and TEO reviewed the paper.}
\\
\noindent
{\small
{\bf Funding}:
Alexander Dack acknowledges funding from the Department of Bioengineering at Imperial College London. Benjamin Qureshi was funded by the European Research Council (ERC) under the European Union’s Horizon 2020 research and innovation program (Grant agreement No. 851910). Thomas E. Ouldridge would like to thank the Royal Society for a University Research Fellowship. Tomislav Plesa would like to thank Peterhouse, 
University of Cambridge, for a Fellowship.}
\\
\noindent
{\small
{\bf Declaration of interests}:
The authors declare no competing interests.
}
\\
\noindent
{\small
{\bf Acknowledgments}:
We thank Chris P. Barnes and Kathleen J. Y. Zhang for useful discussions about synthetic microbial communities.
}
\\
\noindent
{\small
{\bf Data and code availability}:
Data and code to be made available in a public repository with a DOI upon publication. During the review process you may access the following GitHub repository \cite{Dack_github_2024}: \url{https://github.com/alexdack/recurrent_neural_chemical_reaction_networks}.}

\bibliographystyle{ieeetr}  
\bibliography{references.bib}  

\appendix

\section{Appendix: Universal approximation theorem for RNCRNs}
\label{app:single_layer_RNCRN}
In this section, we present the main theoretical result
for single- and multi-layer RNCRNs.

\subsection{Single-layer RNCRN}
Consider the single-layer RNCRN given by
\begin{align}
\varnothing & \xrightarrow[]{\beta_i} X_i,  
\; \; \; \; 
X_i + Y_j \xrightarrow[]{\left|\alpha_{i,j}\right|} X_i + \textrm{sign}(\alpha_{i,j}) X_i + Y_j,\nonumber\\
\varnothing & \xrightarrow[]{\gamma/\mu} Y_j,  
\; \; \; \; 
Y_j \xrightarrow[]{\left|\theta_{j}\right|/\mu} Y_j + \textrm{sign}(\theta_{j}) Y_j, 
\; \; \; \; 
X_i + Y_j \xrightarrow[]{\left|\omega_{j,i}\right|/\mu} X_i + Y_j + \textrm{sign}(\omega_{j,i}) Y_j,
\; \; \; \;  
2 Y_j \xrightarrow[]{1/\mu} Y_j,
\label{eq:CRN_single_layer}
\end{align}
for $i = 1, 2, \ldots, N$ and $j = 1, 2, \ldots, M$,
where for any real number $x \in \mathbb{R}$, we define 
$\textrm{sign}(x) = -1$ if $x < 0$, 
$\textrm{sign}(x) = 0$ if $x = 0$ and
$\textrm{sign}(x) = 1$ if $x > 0$.
The RREs induced by~(\ref{eq:CRN_single_layer}) are 
given by~(\ref{eq:single_layer_RRE}).

In what follows, we define 
$\mathbf{x} = (x_1,x_2,\ldots,x_N) \in \mathbb{R}^N$,
and collect all the rate coefficients 
$\alpha_{i,j}$, $\omega_{j,k}$, $\theta_{j}$  
into suitable vectors 
$\boldsymbol{\alpha} \in \mathbb{R}^{NM}$ ,
$\boldsymbol{\omega} \in \mathbb{R}^{M N}$, $\boldsymbol{\theta} \in \mathbb{R}^{M}$, respectively;
similarly, we let 
$\boldsymbol{\beta} = (\beta_1,\beta_2,\ldots,\beta_N)
\in \mathbb{R}_{\ge}^{N}$.
Furthermore, we define the \emph{reduced vector field} by
\begin{align}
g_i(\mathbf{x}) = 
g_i(\mathbf{x}; \boldsymbol{\alpha}, \boldsymbol{\beta},\gamma,\boldsymbol{\omega}, \boldsymbol{\theta} ) & = \beta_i + 
x_i \sum_{j=1}^M \alpha_{i,j} 
\sigma_{\gamma} \left(\sum_{k=1}^N \omega_{j,k} x_k 
+ \theta_{j} \right),  
\; \; \; \textrm{for } i = 1, 2, \ldots, N,
\label{eq:reduced_functions}
\end{align}
which appears on the 
right-hand side of~(\ref{eq:single_layer_RRE_reduced}).  

\begin{theorem}$($\textbf{Universal approximation theorem for single-layer \emph{RNCRN}s}$)$  \label{the:single_layer_RNCRN}
Consider the target \emph{ODE} system~$(\ref{eq:target_ODEs})$
on a compact set $\mathbb{K} = \mathbb{K}_1 \times \mathbb{K}_2 
\times \ldots \times \mathbb{K}_N \subset \mathbb{R}_{>}^N$ 
in the state-space, with vector field $(f_1,f_2,\ldots,f_N)$ 
Lipschitz-continuous on $\mathbb{K}$. 
Consider also the single-layer 
\emph{RNCRN}~$(\ref{eq:CRN_single_layer})$
with \emph{RREs}~$(\ref{eq:single_layer_RRE})$ 
with rate coefficients $\boldsymbol{\beta} = 
\boldsymbol{\beta}^* \in \mathbb{R}_\geq^{N}$ 
and $\gamma = \gamma^* > 0$ fixed to any values.
\begin{enumerate}
\item[\textbf{\emph{(i)}}] \textbf{Quasi-static approximation}. 
Consider the reduced vector field~$(\ref{eq:reduced_functions})$.
Let $\varepsilon > 0$ be any given tolerance. 
Then, for every sufficiently large $M > 0$ there exist 
$\boldsymbol{\alpha}^* \in \mathbb{R}^{NM}$ , $\boldsymbol{\omega}^* \in \mathbb{R}^{M N}$, and $\boldsymbol{\theta}^* \in \mathbb{R}^{M}$  such that
\begin{align}
\underset{\mathbf{x} \in \mathbb{K}}{\max}
\left| g_i(\mathbf{x}; \boldsymbol{\alpha}^*,
    \boldsymbol{\beta}^*,\gamma^*,\boldsymbol{\omega}^*,\boldsymbol{\theta}^*)
    - f_i(\mathbf{x})\right| \le \varepsilon
\; \; \textrm{for all } i = 1,2,\ldots,N.
\label{eq:quasi_static}
\end{align}
\item[\textbf{\emph{(ii)}}] \textbf{Dynamical approximation}.
Assume that the solution of~$(\ref{eq:target_ODEs})$ 
exists for all $t \in [0,T]$ for some $T > 0$, 
and that $\bar{x}_i(t) \in \mathbb{K}_i$ for all $t \in [0,T]$
for all $i = 1,2,\ldots,N$.
Then, for every sufficiently small $\varepsilon = \varepsilon^* > 0$ fixed
there exists $\mu_0 > 0$ such that for all $\mu \in (0,\mu_0)$
system~$(\ref{eq:single_layer_RRE})$ 
has a unique solution $x_i(t) \in \mathbb{K}_i$ for all $t \in [0,T]$
for all $i = 1,2,\ldots,N$, and 
\begin{align}
\underset{t \in [0,T]}{\max}
\left|x_i(t;\boldsymbol{\alpha}^*,\boldsymbol{\beta}^*,\gamma^*,\boldsymbol{\omega}^*,\boldsymbol{\theta}^*,\mu) 
    - \bar{x}_i(t) \right| & \le c_1 \varepsilon^* + c_2 \mu 
\; \; \textrm{for all } i = 1,2,\ldots,N,
\label{eq:dynamical}
\end{align}
where constants $c_1$ and $c_2$ are independent of $\mu$.
\end{enumerate}
\end{theorem}

\begin{proof} \item[\textbf{(i)}] \textbf{Quasi-static approximation}. 
Consider continuous function $h_i : \mathbb{K} \rightarrow \mathbb{R}$
defined by $h_i(\mathbf{x}) = (f_i(\mathbf{x}) - \beta_i^*)/x_i$.
Since the activation function $\sigma_{\gamma}$, 
defined by~(\ref{eq:smu_activation}), 
is continuous and non-polynomial, it follows 
from~\cite{pinkus_approximation_1999}[Theorem 3.1] that 
for any $\varepsilon > 0$ there exist $M > 0$, 
coefficients $\alpha_{i,j}^*,\omega_{j,k}^*, \theta_{j}^* \in \mathbb{R}$ 
and a continuous function $\rho_i(\mathbf{x})$ such that 
\begin{align}
h_i(\mathbf{x}) & = \sum_{j = 1}^{M} \alpha_{i,j}^* \sigma_{\gamma^*} 
\left(\sum_{k=1}^N \omega_{j,k}^* x_k + \theta_{j}^* \right)
+ \rho_i(\mathbf{x})
\; \; \; \textrm{for all } i = 1,2,\ldots,N,
\label{eq:quasi_static_aux}
\end{align}
and $\max_{\mathbf{x} \in \mathbb{K}} |\rho_i(\mathbf{x})|
\le \varepsilon/X_i$, with $X_i = \max_{x_i \in \mathbb{K}_i} x_i$.
Equation~(\ref{eq:quasi_static_aux}) implies~(\ref{eq:quasi_static}). 

\item[\textbf{(ii)}] \textbf{Dynamical approximation}. 
It follows from~(\ref{eq:quasi_static})
and regular perturbation theory
\cite{coddington_theory_1995}
that there exists $\varepsilon_0 > 0$ such that
for all $\varepsilon \in (0,\varepsilon_0)$
\begin{align}
\underset{t \in [0,T]}{\max}
\left|\tilde{x}_i(t;\boldsymbol{\alpha}^*,
    \boldsymbol{\beta}^*,\gamma^*,\boldsymbol{\omega}^*, \boldsymbol{\theta}^*) 
    - \bar{x}_i(t) \right| & \le c_1 \varepsilon 
\; \; \textrm{for all } i = 1,2,\ldots,N,
\label{eq:regular}
\end{align}
for some $\varepsilon$-independent constant $c_1 > 0$,
where $\tilde{x_i}(t) = \tilde{x}_i(t;\boldsymbol{\alpha}^*,
\boldsymbol{\beta}^*,\gamma^*,\boldsymbol{\omega}^*,\boldsymbol{\theta}^*)$
satisfies the reduced system~(\ref{eq:single_layer_RRE_reduced}).
In what follows, we fix $\varepsilon 
= \varepsilon^* \in (0,\varepsilon_0)$. 

Consider the fast (adjoined) system 
from~(\ref{eq:single_layer_RRE}), defined by
\begin{align}
\frac{\mathrm{d} y_j}{\mathrm{d} \tau} & =  
\gamma^* + \theta_{j}^* y_j
+ \left(\sum_{i=1}^N \omega_{j,i}^* x_i \right) y_j
-  y_j^2, 
\hspace{0.5cm} y_j(0) = b_j, 
\; \; \; \textrm{for } j = 1, 2, \ldots, M,
\label{eq:single_layer_RRE_fast}
\end{align}
where $x_1,x_2,\ldots,x_N$ are parameters.
The ODEs from~(\ref{eq:single_layer_RRE_fast})
are decoupled, and each one has a unique non-negative 
continuously differentiable equilibrium
$y_j = \sigma_{\gamma^*} \left(\sum_{k=1}^N \omega_{j,k}^* 
x_k + \theta_{j}^* \right)$, which is stable
for all non-negative initial conditions $b_j \ge 0$. 
It follows from singular perturbation theory 
(Tikhonov's theorem)~\cite{klonowski_simplifying_1983, Tik52} that there exists $\mu_0 > 0$ such that 
for all $\mu \in (0,\mu_0)$ system~(\ref{eq:single_layer_RRE})
has a unique solution in $\mathbb{K}$ 
over time-interval $[0,T]$, and
\begin{align}
\underset{t \in [0,T]}{\max}
\left|x_i(t; \boldsymbol{\alpha}^*,
    \boldsymbol{\beta}^*,\gamma^*,\boldsymbol{\omega}^*, \boldsymbol{\theta}^*, \mu) 
    - \tilde{x}_i(t; \boldsymbol{\alpha}^*,
    \boldsymbol{\beta}^*,\gamma^*,\boldsymbol{\omega}^*, \boldsymbol{\theta}^*) \right| & \le c_2 \mu 
\; \; \textrm{for all } i = 1,2,\ldots,N,
\label{eq:singular}
\end{align}
for some $\mu$-independent constant $c_2 > 0$. 
Using $|x_i - \bar{x}_i| \le |x_i - \tilde{x}_i|
+ |\tilde{x}_i - \bar{x}_i|$ and~(\ref{eq:regular}) and~(\ref{eq:singular}),
one obtains~(\ref{eq:dynamical}).
\end{proof}

\noindent
\textbf{Remark}. In Theorem~\ref{the:single_layer_RNCRN},
$L_{\infty}$-norm is used to measure the distance between 
the target and reduced vector fields. On the other hand, 
in Algorithm~\ref{algh:RNCRN}, we instead use 
the $L_{2}$-norm, for computational simplicity.
Rigorous justification of using the $L_{2}$-norm is beyond the
scope of this paper.

\subsection{Multi-layer RNCRN}
\label{app:proof_multi_net_multi_dim}
In the ANN context, \emph{deep} neural networks, 
i.e. those containing multiple layer of neurons, 
are important~\cite{lecun_deep_2015}. 
In the CRN context, generalizing the single-layer RREs~(\ref{eq:single_layer_RRE}) to multi-layer ones yields 
\begin{align}
\frac{\mathrm{d} x_i}{\mathrm{d}t} & = \beta_i + x_i\sum_{k=1}^{M_L} \alpha_{i,k} y_{k}^{(L)}, && x_i(0) = a_i, 
\hspace{1.3cm} \textrm{for } i = 1, 2, \ldots, N, \nonumber \\
\frac{\mathrm{d} y_{j_1}^{(1)}}{\mathrm{d}t} & =
\frac{\gamma}{\mu} 
+ \frac{\theta_{j}^{(1)}}{\mu} y_{j_1}^{(1)}
+ \left(\sum_{i=1}^N \frac{\omega_{j,i}^{(1)}}{\mu} x_i\right) y_{j_1}^{(1)} - \frac{1}{\mu} \left(y_{j_1}^{(1)}\right)^2, &&  y_{j_1}^{(1)}(0) = b_{j_1} \geq 0,  \; \; \; \textrm{for } j_1 = 1, 2, \ldots, M_1, \nonumber \\
\frac{\mathrm{d} y_{j_2}^{(2)}}{\mathrm{d}t} & =
\frac{\gamma}{\mu} 
+ \frac{\theta_{j}^{(2)}}{\mu} y_{j_2}^{(2)}
+ \left(\sum_{k=1}^{M_1} \frac{\omega_{j,k}^{(2)}}{\mu} y_k^{(1)}\right) y_{j_2}^{(2)}
- \frac{1}{\mu} \left(y_{j_2}^{(2)}\right)^2, &&  y_{j_2}^{(2)}(0) = b_{j_2} \geq 0,  \; \; \; \textrm{for } j_2 = 1, 2, \ldots, M_2, \nonumber \\
    \vdots \nonumber \\
\frac{\mathrm{d} y_{j_L}^{(L)}}{\mathrm{d}t} &= \frac{\gamma}{\mu} 
+ \frac{\theta_{j}^{(L)}}{\mu} y_{j_L}^{(L)}
+ \left(\sum_{k=1}^{M_{L-1}} \frac{\omega_{j,k}^{(L)}}{\mu} y_k^{(L-1)}\right)  y_{j_L}^{(L)} 
- \frac{1}{\mu} \left(y_{j_L}^{(L)}\right)^2, && y_{j_L}^{(L)}(0) = b_{j_L} \geq 0,  \; \; \; \textrm{for } j_L = 1, 2, \ldots, M_L.
    \label{eq:multi_layer_RNCRN_RRE}
\end{align}
In particular, (\ref{eq:multi_layer_RNCRN_RRE})
involves $N$ executive species $X_1,X_2,\ldots,X_N$, 
and $L$ layers of chemical perceptrons, with the 
$i$th layer containing $M_i$ perceptrons.
In this hierarchy, the executive species 
$X_1,X_2,\ldots,X_N$ are inputs to the first layer
of neurons $Y_1^{(1)}, \dots, Y_{M_1}^{(1)}$, 
which themselves are inputs to the 
second layer of neurons $Y_1^{(2)}, \dots, Y_{M_2}^{(2)}$,
etc.; the final, $L$th, layer of neurons 
$Y_1^{(L)}, \dots, Y_{M_L}^{(L)}$ 
feeds back into the executive species.
The multi-layer RNCRN induced by~(\ref{eq:multi_layer_RNCRN_RRE}) is given by 
\begin{align}
\varnothing & \xrightarrow[]{\beta_i} X_i,  
\; \; \; \; 
X_i + Y_{j_L}^{(L)} \xrightarrow[]{\left|\alpha_{i,{j_L}}\right|} X_i + \textrm{sign}(\alpha_{i,{j_L}}) X_i + Y_{j_L}^{(L)},\nonumber\\
\varnothing & \xrightarrow[]{\gamma/\mu} Y_{j_1}^{(1)},  
\; \; \; \; 
Y_{j_1}^{(1)} \xrightarrow[]{\left|\theta_{j_1}^{(1)}\right|/\mu} Y_{j_1}^{(1)} + \textrm{sign}(\theta_{j_1}^{(1)}) Y_{j_1}^{(1)}
\nonumber\\
X_i + Y_{j_1}^{(1)} & \xrightarrow[]{\left|\omega_{j_1,i}^{(1)}\right|/\mu} X_i + Y_{j_1}^{(1)} + \textrm{sign}(\omega_{j_1,i}^{(1)}) Y_{j_1}^{(1)},
\; \; \; \;  
2 Y_{j_1}^{(1)} \xrightarrow[]{1/\mu} Y_{j_1}^{(1)}, \nonumber \\
\varnothing & \xrightarrow[]{\gamma/\mu} Y_{j_2}^{(2)},  
\; \; \; \; 
Y_{j_2}^{(2)} \xrightarrow[]{\left|\theta_{j_2}^{(2)}\right|/\mu} Y_{j_2}^{(2)} + \textrm{sign}(\theta_{j_2}^{(2)}) Y_{j_2}^{(2)}, 
\nonumber\\
Y_{j_1}^{(1)} + Y_{j_2}^{(2)} & \xrightarrow[]{\left|\omega_{j_2,i}^{(2)}\right|/\mu} Y_{j_1}^{(1)} + Y_{j_2}^{(2)} + \textrm{sign}(\omega_{j_2,i}^{(2)}) Y_{j_2}^{(2)},
\; \; \; \;  
2 Y_{j_2}^{(2)} \xrightarrow[]{1/\mu} Y_{j_2}^{(2)}, \nonumber \\
\vdots \nonumber \\
\varnothing & \xrightarrow[]{\gamma/\mu} Y_{j_L}^{(L)},  
\; \; \; \; 
Y_{j_L}^{(L)} \xrightarrow[]{\left|\theta_{j_L}^{(L)}\right|/\mu} Y_{j_L}^{(L)} + \textrm{sign}(\theta_{j_L}^{(L)}) Y_{j_L}^{(L)}, 
\nonumber\\
Y_{j_{L-1}}^{(L-1)} + Y_{j_L}^{(L)} & \xrightarrow[]{\left|\omega_{j_L,i}^{(L)}\right|/\mu} Y_{j_{L-1}}^{(L-1)} + Y_{j_L}^{(L)} + \textrm{sign}(\omega_{j_L,i}^{(L)}) Y_{j_L}^{(L)},
\; \; \; \;  
2 Y_{j_L}^{(L)} \xrightarrow[]{1/\mu} Y_{j_L}^{(L)}, 
\label{eq:CRN_multi_layer}
\end{align}
for $i = 1, \hdots, N$ and $j_l = 1, \hdots, M_l$ .

\textbf{Reduced ODEs}. Multiplying the RREs of the chemical perceptrons from~(\ref{eq:multi_layer_RNCRN_RRE}) by $\mu$, and then fixing $\mu=0$, it follows that their concentrations are given in the quasi-state state by
\begin{align}
(y_{j_1}^{(1)})^* &= \frac{1}{2} \left(  \left( \sum_{i=1}^N\omega_{j_1,i}^{(1)}x_i + \theta_{j_1}^{(1)} \right) + \sqrt{\left( \sum_{i=1}^N\omega_{j_1,i}^{(1)}x_i + \theta_{j_1}^{(1)} \right)^2  + 4\gamma}\right), \nonumber\\     
\vdots \nonumber \\
(y_{j_L}^{(L)})^* &= \frac{1}{2}\left(\left( \sum_{k=1}^{M_{L-1}}\omega_{j_L,k}^{(L)} (y_k^{(L-1)})^* + \theta_{j_L}^{(L)} \right) + \sqrt{\left( \sum_{k=1}^{M_{L-1}}\omega_{j_L,k}^{(L)} 
(y_k^{(L-1)})^* + \theta_{j_L}^{(L)} \right)^2  + 4\gamma}\right),   
\label{eq:multi_layer_RRE_quasistatic}
\end{align}
for $j_l = 1, 2, \ldots, M_l$.
In this case, the executive species satisfy the reduced ODEs
\begin{align}
\frac{\mathrm{d} x_i}{\mathrm{d}t} &= g_i(x_1,\ldots,x_N)  = \beta_i + x_i\sum_{k=1}^{M_L} \alpha_{i,k} (y_{k}^{(L)})^*, \;\; x_i(0) = a_i \geq 0, \; \; \; \textrm{for } i = 1, 2, \ldots, N.  
    \label{eq:multi_layer_RRE_reduced}
\end{align}
One can formulate analogous theorem to Theorem~\ref{the:single_layer_RNCRN} for the multi-layer RNCRN,
by using~\cite{anderson_reaction_2021}[Definition 4.7 and Proposition 4.9], which ensures that the equilibrium concentration of the 
fast (adjoined) system is globally stable.

\section{Appendix: Examples}
\label{app:code_and_data}
In this section, we present some details
underlying the examples from Section~\ref{sec:examples}.

\subsection{Multi-stable target system}
\label{app:app_tri_stable}

\textbf{Tri-stability}. 
We use Algorithm~\ref{algh:RNCRN} to approximate the target system~(\ref{eq:multistable_ex}) on $\mathbb{K}_1 = [1,18]$. Tolerance
 $\varepsilon \approx 0.5 \times 10^{-3}$ is met with an RNCRN
 with $M = 4$ perceptrons, coefficients 
 $\beta_1 = 0$, $\gamma=1$, and
 \begin{align}
    \boldsymbol{\alpha}_1 =  \begin{pmatrix}
        -4.247 \\ 0.487 \\ 1.363 \\ 0.185
    \end{pmatrix}, \;\; \boldsymbol{\theta} = \begin{pmatrix}
     -4.968\\  -3.590 \\ -11.236\\ 10.249 
    \end{pmatrix},\;\; \boldsymbol{\omega}_{1} = \begin{pmatrix}
     0.511 \\ 1.043 \\ 1.078 \\ -2.355
    \end{pmatrix},
    \label{eq:tri_stable_coeff}
\end{align}
where $\boldsymbol{\alpha}_1 = (\alpha_{1,1},\alpha_{1,2},\alpha_{1,3},
\alpha_{1,4})^{\top}$, $\boldsymbol{\theta} = 
(\theta_{1},\theta_{2},\theta_{3},\theta_{4})^{\top}$, and $\boldsymbol{\omega}_1 = 
(\omega_{1,1},\omega_{2,1},\omega_{3,1},\omega_{4,1})^{\top}$.
The reduced ODE is given by 
\begin{align}
\frac{\mathrm{d} \tilde{x}_1}{\mathrm{d} t} & = 
g_1(\tilde{x}_1) = \tilde{x}_1 \sum_{j=1}^{4}\alpha_{1,j} 
\sigma_{1} \left(\omega_{j,1} \tilde{x}_1 + \theta_{j} \right),
\label{eq:single_layer_RRE_reduced_sin_tri_stable}
\end{align}
 while the full ODEs read
 \begin{align}
\frac{\mathrm{d} x_1}{\mathrm{d}t} &= x_1\left(\sum_{j=1}^4 \alpha_{1,j}y_j\right), && x_1(0)= a_1 \in  \mathbb{K}_1,\nonumber \\
\frac{\mathrm{d} y_j}{\mathrm{d}t} &=  \frac{1}{\mu} + 
\frac{\theta_{j}}{\mu} y_j + \frac{\omega_{j,1}}{\mu} x_1 y_j - 
\frac{1}{\mu} y_j^2, && y_j(0)= b_j \geq 0, 
\; \; \; j = 1,2,\ldots,4.
\label{eq:single_layer_RRE_full_sin_tri_stable}
\end{align}

\textbf{Multi-stability}.
Analogous to Figure~\ref{fig:bi_stable_sin_dynamics_and_traj},
we present some results for uni-, bi-, tri-, quad- and penta-stable 
RNCRNs in Figure~\ref{fig:multi_stable_scale}(a)--(e). 
Furthermore, we present in Figure~\ref{fig:multi_stable_scale}(f)
the number of chemical perceptrons in the RNCRNs as a function
of the number of the underlying stable equilibria, showing a linear
trend; in particular, this plot suggests that $n + 1$ chemical perceptrons may be sufficient to achieve $n$ stable equilibria.
However, we note that the RNCRNs underlying Figure~\ref{fig:multi_stable_scale} are not necessarily 
the smallest ones achieved the indicated multi-stability.

\begin{figure}[ht]
    \centering
    \includegraphics[width=\columnwidth]{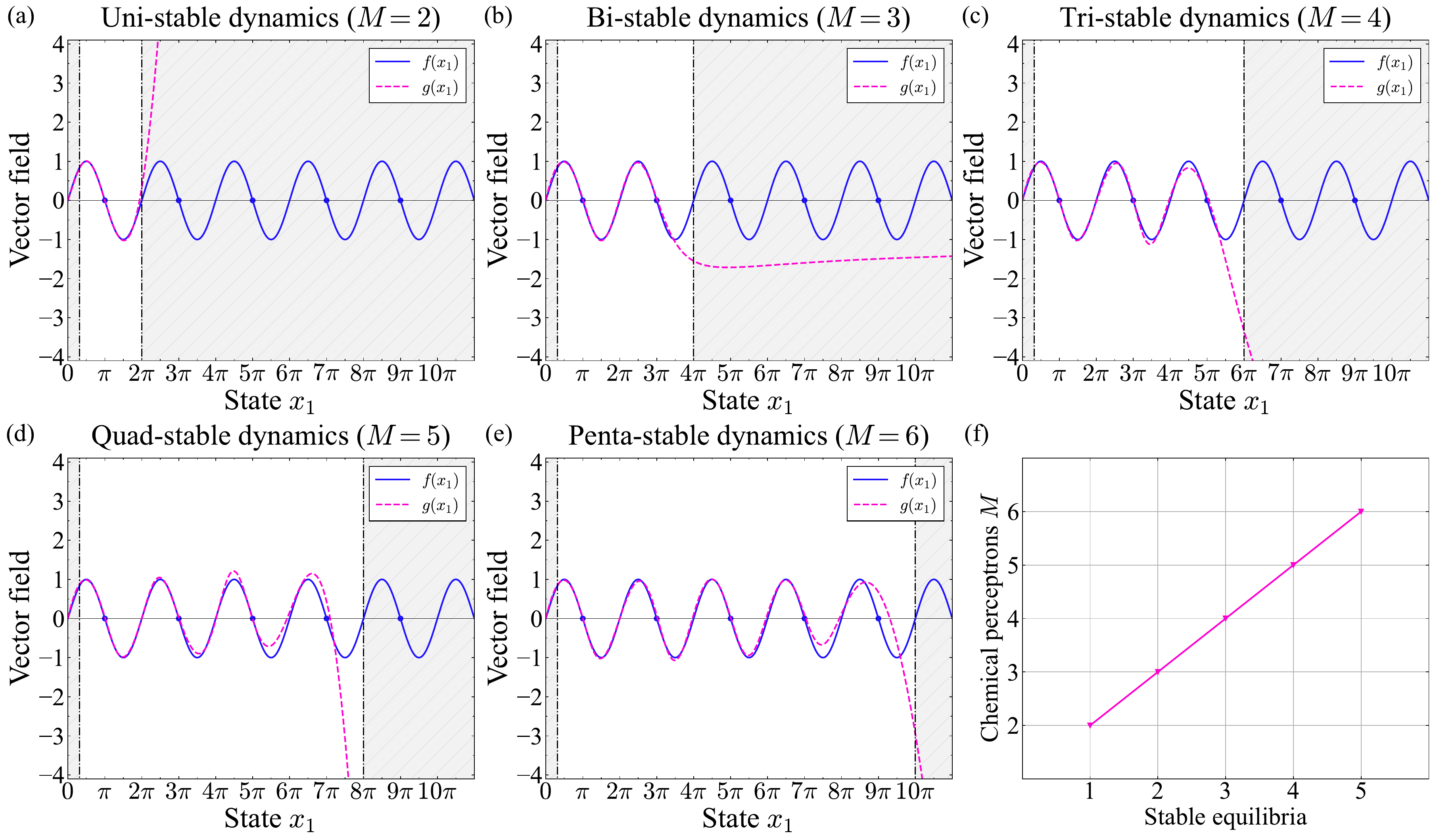}
    \caption{RNCRN approximations of the multi-stable target system~(\ref{eq:multistable_ex}). (a)-(e) The vector field of the target system~(\ref{eq:multistable_ex}) and reduced systems for some 
    RNCRNs trained over the following state-space sets: 
    $\mathbb{K}_1 = [1,2\pi]$ for uni-stability, 
    $\mathbb{K}_1 = [1,4\pi]$ for bi-stability, 
    $\mathbb{K}_1 = [1,6\pi]$ for tri-stability,
    $\mathbb{K}_1 = [1,8\pi]$ for quad-stability,  
    and $\mathbb{K}_1 = [1,10\pi]$ for penta-stability. 
    (f) Number of chemical perceptrons in the RNCRNs 
    as a function of the number of the underlying stable equilibria. 
    The RNCRNs are available in the code repository \cite{Dack_github_2024}.}
    \label{fig:multi_stable_scale}
\end{figure}

\subsection{Oscillatory target system}
\label{sec:app_bessel_osc}
We use Algorithm~\ref{algh:RNCRN} to approximate the target system~(\ref{eq:bessel_osc}) on $\mathbb{K}_1
\times \mathbb{K}_2 = [1,8] \times [2, 10]$. Tolerance $\varepsilon \approx 10^{-4}$ is met with an RNCRN with $M=6$ perceptrons and coefficients $\beta_1=\beta_2 = 1$, $\gamma = 1$, and  
\begin{align}
    \boldsymbol{\alpha}_1 &= \begin{pmatrix}
    0.664 \\
    3.392 \\
    -1.159 \\
    10.658 \\
    -2.799 \\
    -0.571
    \end{pmatrix}, \; \; 
    \boldsymbol{\alpha}_2 = \begin{pmatrix} 
    4.260 \\
    0.273 \\
    -0.034 \\
    0.134 \\
    -6.479 \\
    -0.028
    \end{pmatrix}, \;\;
    \boldsymbol{\theta} &= \begin{pmatrix}
        -1.833\\
        -2.452\\
        8.979\\	
        2.420\\
        0.344\\
        -6.624
    \end{pmatrix}, \;\;
    \boldsymbol{\omega}_1 = \begin{pmatrix}
0.290 \\	0.466\\	-2.195\\	-1.584\\	-0.2900\\	2.182
    \end{pmatrix}, \;\;
    \boldsymbol{\omega}_2 = \begin{pmatrix}
-0.572\\	0.112\\	-0.015\\	-0.226\\	-0.716\\	0.788
    \end{pmatrix},
    \label{eq:coefficents_oscill}
\end{align}
where $\boldsymbol{\alpha}_i = (\alpha_{i,1},\alpha_{i,2},\ldots, \alpha_{i,6})^{\top}$ for $i = 1,2$, 
$\boldsymbol{\theta} = 
(\theta_{1},\theta_{2},\theta_{3},\theta_{4},\theta_{5}, \theta_{6})^{\top}$, and  $\boldsymbol{\omega}_k = 
(\omega_{1,k},\omega_{2,k},\ldots,\omega_{6,k})^{\top}$
for $k = 1, 2$.
The reduced ODEs are given by 
\begin{align}
\frac{\mathrm{d} \tilde{x}_1}{\mathrm{d} t} & = 
g_1(\tilde{x}_1, \tilde{x}_2) = 1 + \tilde{x}_1 \sum_{j=1}^{6}\alpha_{1,j} 
\sigma_{1} \left(\omega_{j,1} \tilde{x}_1 + \omega_{j,2} \tilde{x}_2 + \theta_{j} \right), \nonumber \\
\frac{\mathrm{d} \tilde{x}_2}{\mathrm{d} t} & = 
g_2(\tilde{x}_1, \tilde{x}_2) = 1 + \tilde{x}_2 \sum_{j=1}^{6}\alpha_{2,j} 
\sigma_{1} \left(\omega_{j,1} \tilde{x}_1 + \omega_{j,2} \tilde{x}_2 + \theta_{j} \right), 
\label{eq:single_layer_RRE_reduced_oscill}
\end{align}
 while the full ODEs read
\begin{align}
\frac{\mathrm{d} x_1}{\mathrm{d}t} 
& = 1+ x_1\left(\sum_{j=1}^6 \alpha_{1,j}y_j\right), 
&& x_1(0)= a_1 \in \mathbb{K}_1, \nonumber \\
\frac{\mathrm{d} x_2}{\mathrm{d}t} 
& = 1 + x_2\left(\sum_{j=1}^6 \alpha_{2,j}y_j\right), 
&& x_2(0)= a_2 \in \mathbb{K}_2, \nonumber \\
\frac{\mathrm{d} y_j}{\mathrm{d}t} &= \frac{1}{\mu} 
+ \frac{\theta_{j}}{\mu} y_j
+ \left(\sum_{i=1}^2 \frac{\omega_{j,i}}{\mu} x_i \right) y_j 
- \frac{1}{\mu} y_j^2, 
&& y_j(0)= b_j \geq 0, 
\; \; \; j = 1,2,\ldots,6.
\label{eq:full_rncrn_oscil}
\end{align}
The coefficients were rounded 
to 3 decimal places before being used in simulations.

\subsection{Chaotic target system}
\label{sec:hasting_powell_app}
We use Algorithm~\ref{algh:RNCRN} to approximate the target system~(\ref{eq:hasting_powell_ode}) on $x_1 \in \mathbb{K}_1$, 
$x_2 \in \mathbb{K}_2$ and $x_2 \in \mathbb{K}_3$
with $\mathbb{K}_1 = \mathbb{K}_2 = \mathbb{K}_3 = [0.01,1]$. Tolerance $\varepsilon \approx 10^2$ is met with an RNCRN with $M=5$ perceptrons and coefficients $\beta_1=\beta_2 = \beta_3 = 0$, $\gamma = 1$, and  
\begin{align}
    \boldsymbol{\alpha}_1 &= \begin{pmatrix}
         -0.272 \\ 2.996 \\  0.862 \\ -0.244 \\1.276 \\
    \end{pmatrix}, \; \;
     \boldsymbol{\alpha}_2 = \begin{pmatrix}
      0.109 \\ 24.039 \\ -5.668  \\ -0.057 \\ -9.584 \\
    \end{pmatrix}, \; \;    
    \boldsymbol{\alpha}_3 = \begin{pmatrix}
     0.026 \\ -0.529 \\  -1.101 \\ 0.034 \\ 1.065
    \end{pmatrix}, \; \; \nonumber \\   
    \boldsymbol{\theta} &= \begin{pmatrix}
    0.284 \\
    -1.589 \\
    -0.178 \\
    1.212 \\
     -0.707
    \end{pmatrix}, \;\;
    \boldsymbol{\omega}_{1} = \begin{pmatrix}
    -5.049 \\
     0.148\\
     0.506 \\
     15.973 \\
     -1.151 \\
    \end{pmatrix}, \;\;
    \boldsymbol{\omega}_{2} = \begin{pmatrix}
    8.895  \\
     -2.951  \\
     -4.504  \\
      -7.781 \\
      -0.606  \\
    \end{pmatrix}, \;\;
    \boldsymbol{\omega}_{3} = \begin{pmatrix}
    -0.068 \\
     -0.525  \\
     0.329   \\
     -0.027 \\
      0.199 \\
    \end{pmatrix},
    \label{eq:coeff_chaos}
\end{align}
where $\boldsymbol{\alpha}_i = (\alpha_{i,1},\alpha_{i,2},\ldots, \alpha_{i,5})^{\top}$ for $i = 1,2,3$, $\boldsymbol{\theta} = 
(\theta_{1},\theta_{2},\theta_{3},\theta_{4},\theta_{5})^{\top}$, and   
$\boldsymbol{\omega}_k = 
(\omega_{1,k},\omega_{2,k},\ldots,\omega_{5,k})^{\top}$
for $k = 1, 2, 3$.
The reduced ODEs are given by 
\begin{align}
\frac{\mathrm{d} \tilde{x}_i}{\mathrm{d} t} & = 
g_i(\tilde{x}_1, \tilde{x}_2, \tilde{x}_3) = \tilde{x}_i \sum_{j=1}^{5}\alpha_{i,j} 
\sigma_{1} \left(\omega_{j,1} \tilde{x}_1 + \omega_{j,2} \tilde{x}_2 + \omega_{j,3} \tilde{x}_3 + \theta_{j} \right),
\; \; \; i = 1,2,3,
\label{eq:single_layer_RRE_reduced_chaos}
\end{align}
 while the full ODEs read
\begin{align}
\frac{\mathrm{d} x_i}{\mathrm{d}t} &= x_i\left(\sum_{j=1}^5 \alpha_{i,j} y_j\right), && x_i(0)= a_i \in \mathbb{K}_i, 
\; \; i = 1,2,3, \nonumber \\
\frac{\mathrm{d} y_j}{\mathrm{d}t} &= \frac{1}{\mu} 
+ \frac{\theta_{j}}{\mu} y_j + 
\left( \sum_{i=1}^3 \frac{\omega_{j,i}}{\mu} x_i \right) y_j 
- \frac{1}{\mu} y_j^2, 
&& y_j(0)= b_j \geq 0, 
\; \; \; j = 1,2,\ldots,5.
 \label{eq:full_rncrn_chaos}
\end{align}
The coefficients were rounded to 3 decimal 
places before being used in simulations.

\section{Appendix: DNA-strand-displacement implementation}
\label{sec:implementation} 
Let us consider the
RREs~(\ref{eq:single_layer_RRE_full_sin_bi_stable}), and assume now that 
the concentrations $x_1, y_1, y_2, y_3$ are measured
in units of moles (M), while time $t$ in units of seconds (s).
In this section, we compile the induced RNCRN
into a DNA-strand-displacement system 
with physically admissible rate coefficients 
(i.e. rate coefficients between 
$1 M^{-1} s^{-1}$ to $10^6 M^{-1} s^{-1}$)
and admissible concentrations (i.e. concentrations less 
than $10^{-5} M$)~\cite{soloveichik_dna_2010};
to ensure these constraints, we fix the percepton speed
at a larger value, $\mu = 0.5$ in~(\ref{eq:single_layer_RRE_full_sin_bi_stable}). 
Such DNA systems rely on suitable ``fuel'' species, 
which are depleted over time, being present 
at sufficiently high concentrations. 
To simplify analysis and simulations, 
one approach is to assume that the concentrations 
of the fuel species are held constant.
This \emph{constant-fuel assumption}~\cite{soloveichik_dna_2010} can be maintained \textit{in situ} by additional experimental techniques that produce components\cite{bae_situ_2021, schaffter_cotranscriptionally_2022}, increased buffering \cite{scalise_dna_2018}, or manually replenishing the fuel species at regular time-intervals.
In the remainder of this section, we follow three steps:
firstly, we introduce a suitable rescaled RNCRN;
secondly, we map this network into a DNA-based RNCRN,
and choose the rescaling such that its rate coefficients
and species concentrations are admissible;
finally, we simplify the latter network by making
the constant-fuel assumption.

\textbf{Rescaled RNCRN}. Let us first rescale 
time and concentrations from RREs~(\ref{eq:single_layer_RRE_full_sin_bi_stable})
by defining the new variables
$\tau = \kappa t$, $\bar{\bar{x}}(t) = \sigma x(t)$, 
$\bar{\bar{y}}_i(t) = \sigma y_i(t)$ for $i = 1,2,3$~\cite{soloveichik_dna_2010}, 
where $\kappa, \sigma > 0$ are dimensionless, 
and fix the dimensionless parameter $\mu = 0.5$, leading to the \emph{rescaled} RREs:
\begin{align}
\frac{\mathrm{d} \bar{\bar{x}}_1}{\mathrm{d} \tau} & = 
-\frac{0.983}{\kappa \sigma}  \bar{\bar{x}}_1 \bar{\bar{y}}_1 
- \frac{0.050}{\kappa \sigma}  \bar{\bar{x}}_1 \bar{\bar{y}}_2 
+ \frac{2.398}{\kappa \sigma}  \bar{\bar{x}}_1 \bar{\bar{y}}_3,
 && \bar{\bar{x}}_1(0) = \sigma a_1, \nonumber\\
\frac{\mathrm{d} \bar{\bar{y}}_1}{\mathrm{d} \tau} & = 
2 \frac{\sigma}{\kappa} + \frac{15.578}{\kappa} \bar{\bar{y}}_1 
- \frac{2.334}{\kappa \sigma} \bar{\bar{x}}_1 \bar{\bar{y}}_1
- \frac{2}{\kappa \sigma} \bar{\bar{y}}_1^2, 
&& \bar{\bar{y}}_1(0) = \sigma b_1, \nonumber \\
\frac{\mathrm{d} \bar{\bar{y}}_2 }{\mathrm{d} \tau} & = 
2 \frac{\sigma}{\kappa} - \frac{3.836}{\kappa} \bar{\bar{y}}_2
+ \frac{1.988}{\kappa \sigma} \bar{\bar{x}}_1 \bar{\bar{y}}_2
- \frac{2}{\kappa \sigma} \bar{\bar{y}}_2^2, 
&& \bar{\bar{y}}_2(0) = \sigma b_2, \nonumber \\
\frac{\mathrm{d}  \bar{\bar{y}}_3}{\mathrm{d} \tau} & = 
2 \frac{\sigma}{\kappa} + \frac{7.148}{\kappa}  \bar{\bar{y}}_3 
- \frac{1.460}{\kappa \sigma} \bar{\bar{x}}_1  \bar{\bar{y}}_3
- \frac{2}{\kappa \sigma}  \bar{\bar{y}}_3^2, 
&&  \bar{\bar{y}}_3(0) = \sigma b_3.
\label{eq:single_layer_RRE_full_sin_bi_stable_rescaled}
\end{align}
Denoting the chemical species with the rescaled
concentrations $\bar{\bar{x}}_1$, $\bar{\bar{y}}_i$
by $\bar{\bar{X}}_1$ and $\bar{\bar{Y}}_i$, respectively,
the \emph{rescaled} RNCRN induced by~(\ref{eq:single_layer_RRE_full_sin_bi_stable_rescaled}) 
is given by 
\begin{align}
&\bar{\bar{X}}_1 + \bar{\bar{Y}}_1\xrightarrow{0.983/(\kappa \sigma)} \bar{\bar{Y}}_1, 
&&\bar{\bar{X}}_1 + \bar{\bar{Y}}_2\xrightarrow{0.050/(\kappa \sigma)} \bar{\bar{Y}}_2,
&&&\bar{\bar{X}}_1 + \bar{\bar{Y}}_3\xrightarrow{2.398/(\kappa \sigma)} 
2 \bar{\bar{X}}_1 + \bar{\bar{Y}}_3,\nonumber \\
&\bar{\bar{X}}_1 + \bar{\bar{Y}}_1\xrightarrow{2.334/(\kappa \sigma)} 
\bar{\bar{X}}_1,
&&\bar{\bar{X}}_1 + \bar{\bar{Y}}_2\xrightarrow{1.988/(\kappa \sigma)} 
\bar{\bar{X}}_1 + 2 \bar{\bar{Y}}_2,
&&&\bar{\bar{X}}_1 + \bar{\bar{Y}}_3\xrightarrow{1.460/(\kappa \sigma)} \bar{\bar{X}}_1, \nonumber \\
&\bar{\bar{Y}}_1\xrightarrow{15.578/\kappa} 2 \bar{\bar{Y}}_1,
&&\bar{\bar{Y}}_2\xrightarrow{3.836/\kappa} \varnothing,
&&&\bar{\bar{Y}}_3\xrightarrow{7.148/\kappa} 2 \bar{\bar{Y}}_3,\nonumber \\
&\varnothing \xrightarrow{2.000 \sigma/\kappa} \bar{\bar{Y}}_1,
&&\varnothing \xrightarrow{2.000 \sigma/\kappa} \bar{\bar{Y}}_2,
&&&\varnothing \xrightarrow{2.000 \sigma/\kappa} \bar{\bar{Y}}_3,\nonumber \\
&2 \bar{\bar{Y}}_1\xrightarrow{2.000/(\kappa \sigma)} \bar{\bar{Y}}_1,
&&2 \bar{\bar{Y}}_2 \xrightarrow{2.000/(\kappa \sigma)} \bar{\bar{Y}}_2,
&&&2 \bar{\bar{Y}}_3 \xrightarrow{2.000/(\kappa \sigma) }\bar{\bar{Y}}_3.
\label{eq:bistable_crn}
\end{align}

\textbf{DNA-RNCRN}. The rescaled RNCRN~(\ref{eq:bistable_crn})
can be mapped into a DNA-strand-displacement-based CRN~\cite{soloveichik_dna_2010},
which we call the DNA-RNCRN, and present in
the text files entitled ``domain\_level\_DNA\_full\_inits''
in the folder ``ex\_5\_dna\_implementation/CRNs''
available in~\cite{Dack_github_2024}.
This network consists of $78$ chemical species
(representing single- and double-stranded DNAs), 
and $53$ reactions, all of which involve
exactly two reactants. We fix
the scaling parameters to $\kappa = 3 \times 10^{4}$
and $\sigma = 5 \times 10^{-9}$; then, 
the underlying rate coefficients of all the
chemical reactions in the DNA-RNCRN 
can be shown to be between $7 M^{-1} s^{-1}$ 
and $10^{6} M^{-1} s^{-1}$, and the concentrations
of all of the species are then less than
$10^{-5} M$, i.e. within the physically admissible range.
We show the concentration of the executive species $\bar{\bar{x}}_1(t)$ 
from the DNA-RNCRN for various initial conditions 
as the red curves in Figure~\ref{fig:bi_stable_dna_implementation_full}(a);
also shown for comparison is the corresponding solution
from the rescaled target system~(\ref{eq:single_layer_RRE_full_sin_bi_stable_rescaled}).
One can notice that the DNA-RNCRN approximately displays the intended bistability for some interval of time. However, as time progresses
further, some of the fuel species are significantly depleted and, consequently, the DNA-RNCRN starts to deviate significantly from its target behavior.

\begin{figure}[H]
\centering
\includegraphics[align=t, width=0.69\columnwidth ]{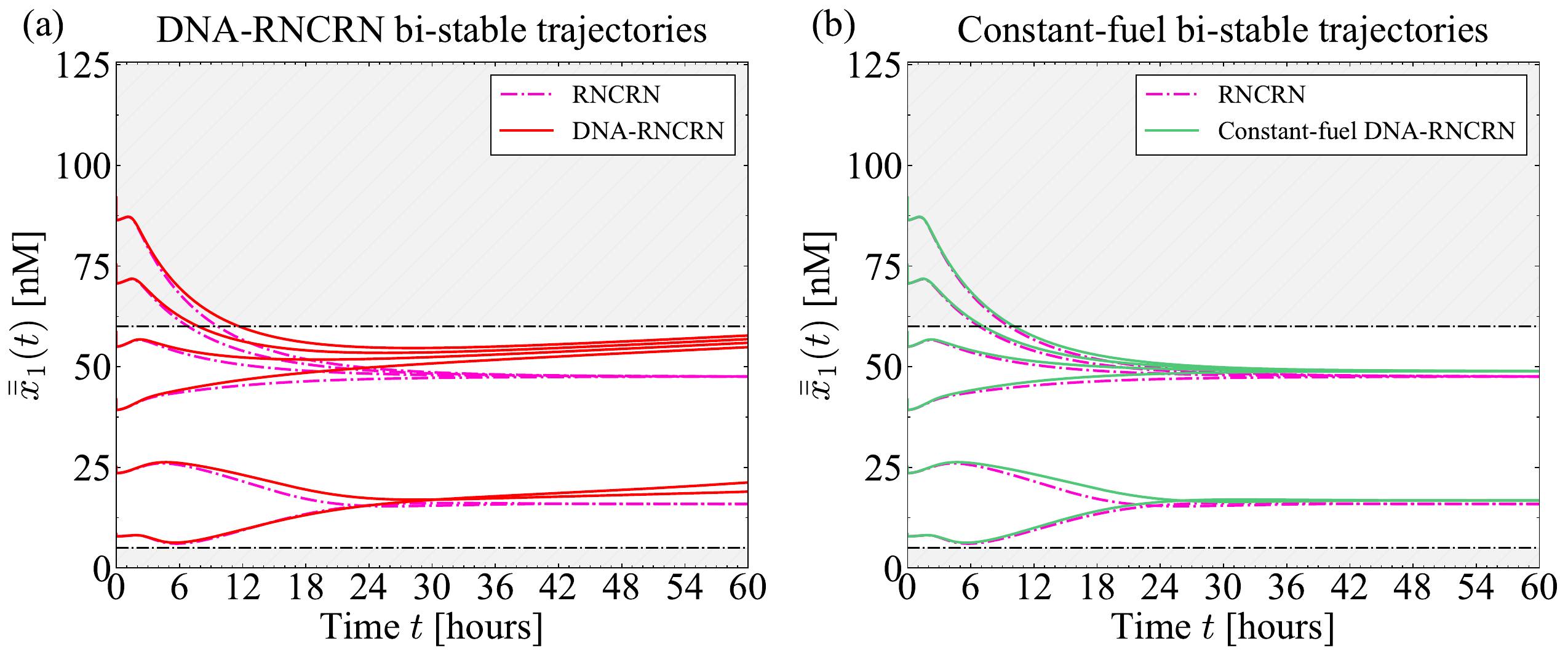}
\caption{DNA-strand-displacement implementation of the RNCRN
with RREs~(\ref{eq:single_layer_RRE_full_sin_bi_stable})
with $\mu = 0.5$. 
Panel (a) shows in red the concentration of the executive species
from the DNA-RNCRN, presented in~\cite{Dack_github_2024}.
Also shown in purple is the concentrations 
of the executive species from the RREs~(\ref{eq:single_layer_RRE_full_sin_bi_stable_rescaled}), 
with initial conditions $\bar{\bar{x}}_1(0) 
\in \{\SI{7.853981633974483e-09}{M}, 
\SI{2.356194490192345e-08}{M}, \SI{3.9269908169872415e-08}{M},  \SI{5.497787143782138e-08}{M}, \SI{7.068583470577035e-08}{M}, \SI{8.63937979737193e-08}{M}\}$, and zero initial concentration
for all chemical perceptrons.
Analogous plot is shown in panel (b), where
the green curve is the concentration of the executive
species from the constant-fuel RNCRN~(\ref{eq:dna_simple_rncrn}),
with initial conditions 
$\bar{\bar{x}}_{1}(0) \in \{\SI{8.368001073741345e-09}{M}, 
\SI{2.5104003221224034e-08}{M}, \SI{4.1840005368706726e-08}{M},  \SI{5.857600751618941e-08}{M}, \SI{7.53120096636721e-08}{M}, \SI{9.204801181115478e-08}{M}\}$, $C_{11}(0) = C_{12}(0)= C_{13}(0) = \SI{5.327234938736891e-09}{M}$, and all the other species
having zero concentration initially. Note that, due to 
the buffering modules, initial conditions for the executive
species from~(\ref{eq:bistable_crn}) and~(\ref{eq:dna_simple_rncrn})
are slightly different. The simulations were preformed using the full precision stated in the associated text files.}
    \label{fig:bi_stable_dna_implementation_full}
\end{figure} 

In what follows, we outline how the DNA-RNCRN is designed;
we abuse the notation slightly,
by following the notation from~\cite{soloveichik_dna_2010}.

\emph{Zero-order reactions}. 
The abstract zero-order reaction (i.e. reaction with $\varnothing$
as the reactants) of the form 
$\varnothing \xrightarrow{\gamma_i} \bar{\bar{Y}}_i$
is mapped to the DNA-based network given by 
\begin{align} 
C_i + G_i \xrightarrow{10^5 \gamma_i/\delta} O_i, \nonumber \\
 O_i + T_i \xrightarrow{10^6} C_i + \bar{\bar{Y}}_i,
\end{align}
where $C_i$, $\bar{\bar{Y}}_i$, and $O_i$ are single-stranded DNAs,
while $G_i$ and $T_i$ are double-stranded DNA complexes
called gates, acting as fuel species.
Let us note that $C_i$ plays an effectively catalytic role
in this reaction cascade.
The initial concentrations for the
fuel species $G_i$ and $T_i$ are set to respectively
$G_i(0) = 10^{-5} M$ and $T_i(0)= 10^{-5} M$.
Furthermore, to ensure the effective reaction rate $\gamma_i$, 
we fix the initial condition of $C_i$ to 
$C_i(0) = \delta M$; in particular, 
in the DNA-strand-displacement framework, effective rates of
reactions are calibrated using the initial conditions of suitable DNA species.

\emph{First-order reactions}. The abstract first-order reaction (i.e. reaction with only one reactant) of the form 
$\bar{\bar{Y}}_i \xrightarrow{\theta_i} 2 \bar{\bar{Y}}_i$
is mapped to the DNA-based network given by 
\begin{align} 
\bar{\bar{Y}}_i + G_i \xrightarrow{10^5 \theta_i} O_i, \nonumber \\
O_i + T_i \xrightarrow{10^6} 2 \bar{\bar{Y}}_i,
\end{align}
where $\bar{\bar{Y}}_i$ and $O_i$ are single-stranded DNAs and, 
as before, $G_i$ and $T_i$ are fuel species with initial concentrations 
$G_i(0) = T_i(0)= 10^{-5} M$. 
A similar DNA-based network is designed for 
abstract reactions of the form 
$Y_i \xrightarrow{\theta_i} \varnothing$.

\emph{Second-order reactions}. Finally, 
the abstract second-order reaction 
(i.e. reaction with two reactants) of the form 
$2 \bar{\bar{Y}}_i \xrightarrow{\kappa_i} \bar{\bar{Y}}_i$
is mapped to the DNA-based network given by 
\begin{align} 
\bar{\bar{Y}}_i + L_i  \xrightleftharpoons[10^6]{\kappa_i} 
H_i + B_i, \nonumber \\
\bar{\bar{Y}}_i + H_i \xrightarrow{10^6} O_i,  \nonumber \\
O_i + T_i \xrightarrow{10^6} \bar{\bar{Y}}_i, 
\end{align}
where $\bar{\bar{Y}}_i$, $B_i$, and $O_i$ are single-stranded,
while $L_i$, $H_i$, $T_i$ double-stranded DNAs. 
The initial concentrations
for the fuel species are set to 
$L_i(0) = T_i(0)= B_i(0)= 10^{-5} M$. Similar DNA-based networks 
are designed for the other variants of second-order reactions.

\emph{Buffering}. We also introduce into the 
DNA-RNCRN the so-called 
buffering modules~\cite{soloveichik_dna_2010}.
These modules require that certain rate coefficients
and initial conditions are multiplied by a factor,
which we account for in the presented results.

\textbf{Constant-fuel DNA-RNCRN}. 
To address the deviation of the DNA-RNCRN from its target behavior displayed in Figure~\ref{fig:bi_stable_dna_implementation_full}(a), 
the fuel species must be replenished.
In this paper, we do not model such replenishment 
explicitly; instead, for simplicity, we 
assume that the concentrations of 
all fuel species in the DNA-RNCRN are constant.
Then, this network can be approximated~\cite{soloveichik_dna_2010} 
by the \emph{constant-fuel} DNA-RNCRN, given by 
\begin{align}&\bar{\bar{Y}}_{2} \xrightleftharpoons[\num{10}]{\num{0.5124102791074654}}HS_{2} ,&&\bar{\bar{Y}}_{1} \xrightleftharpoons[\num{10}]{\num{0.5124102791074654}}HS_{1} ,&&&\bar{\bar{Y}}_{3} \xrightleftharpoons[\num{10}]{\num{0.5124102791074654}}HS_{3}, \nonumber \\ 
&\bar{\bar{X}}_{1} \xrightleftharpoons[\num{10}]{\num{0.0698546537735527}}H_{1} ,&&\bar{\bar{Y}}_{1} +H_{1} \xrightarrow{\num{1.00E+06}}\bar{\bar{Y}}_{1} ,&&&\bar{\bar{X}}_{1} \xrightleftharpoons[\num{10}]{\num{0.0035581445989837993}}H_{2}, \nonumber \\ 
&\bar{\bar{Y}}_{2} + H_{2} \xrightarrow{\num{1.00E+06}}\bar{\bar{Y}}_{2}, &&\bar{\bar{X}}_{1} \xrightleftharpoons[\num{10}]{\num{0.17034970193398138}}H_{3} ,&&&\bar{\bar{Y}}_{3} +H_{3} \xrightarrow{\num{1.00E+06}} 2 \bar{\bar{X}}_{1} +\bar{\bar{Y}}_{3}, \nonumber \\ 
&\bar{\bar{X}}_{1} \xrightleftharpoons[\num{10}]{\num{0.1657846429799732}}H_{4}, &&\bar{\bar{Y}}_{1} +H_{4} \xrightarrow{\num{1.00E+06}}\bar{\bar{X}}_{1},&&&\bar{\bar{X}}_{1} \xrightleftharpoons[\num{10}]{\num{0.14119049474173975
}}H_{5}, \nonumber \\ 
&\bar{\bar{Y}}_{2} +H_{5} \xrightarrow{\num{1.00E+06}} 
\bar{\bar{X}}_{1} + 2 \bar{\bar{Y}}_{2}, &&\bar{\bar{X}}_{1} \xrightleftharpoons[\num{10}]{\num{0.10373223944555172}}H_{6}, &&&\bar{\bar{Y}}_{3} + H_{6} \xrightarrow{\num{1.00E+06}} \bar{\bar{X}}_{1}, \nonumber \\ 
&\bar{\bar{Y}}_{1} \xrightarrow{\num{0.0005532782039943088}}O_{7} ,&&O_{7} \xrightarrow{\num{10}} 2 \bar{\bar{Y}}_{1}, &&&\bar{\bar{Y}}_{2} \xrightarrow{\num{0.00013624981070398546
}}O_{8}, \nonumber \\ 
&\bar{\bar{Y}}_{3} \xrightarrow{\num{0.00025382830078090154}}O_{9} ,&&O_{9} \xrightarrow{\num{10}} 2 \bar{\bar{Y}}_{3},&&& C_{11} \xrightarrow{\num{7.102979918315853e-05}}O_{11}, \nonumber \\ 
& O_{11} \xrightarrow{\num{10}} C_{11} +\bar{\bar{Y}}_{1}, 
&& C_{12} \xrightarrow{\num{7.102979918315853e-05}} O_{12},&&& O_{12} \xrightarrow{\num{10}} C_{12} +\bar{\bar{Y}}_{2}, 
\nonumber \\ 
&C_{13} \xrightarrow{\num{7.102979918315853e-05}}O_{13},&& O_{13} \xrightarrow{\num{10}} C_{13} +\bar{\bar{Y}}_{3}, &&&\bar{\bar{Y}}_{1} \xrightleftharpoons[\num{10}]{\num{0.1420595983663171}}H_{14}, \nonumber \\ 
&\bar{\bar{Y}}_{1} +H_{14} \xrightarrow{\num{1.00E+06}}\bar{\bar{Y}}_{1} ,&&\bar{\bar{Y}}_{2} \xrightleftharpoons[\num{10}]{\num{0.1420595983663171
}}H_{15}, &&&\bar{\bar{Y}}_{2} + H_{15} \xrightarrow{\num{1.00E+06}}\bar{\bar{Y}}_{2}, \nonumber \\ 
&\bar{\bar{Y}}_{3} \xrightleftharpoons[\num{10}]{\num{0.1420595983663171}}H_{16}, &&\bar{\bar{Y}}_{3} +H_{16} \xrightarrow{\num{1.00E+06}}\bar{\bar{Y}}_{3}, 
\label{eq:dna_simple_rncrn} \end{align}
with $C_{11}(0) = C_{12}(0)= C_{13}(0) = \SI{5.327234938736891e-09}{M}$, 
while the initial conditions for all the other 
auxiliary species (i.e. all species except 
the executive species $X_1$ 
and chemical pereptrons $Y_1, Y_2, Y_3$)
are set to zero. Let us note that the fuel species
concentrations from the DNA-RNCRN, 
assumed to be fixed to constant $10^{-5} M$, are absorbed 
into effective rate coefficients displayed above
the reaction arrows in~(\ref{eq:dna_simple_rncrn}).
The concentration of the executive species $\bar{\bar{x}}_1(t)$
from~(\ref{eq:dna_simple_rncrn}) is shown in  
Figure~\ref{fig:bi_stable_dna_implementation}(b), 
which we repeat for convenience as Figure~\ref{fig:bi_stable_dna_implementation_full}(b).
One can observe that the constant-fuel DNA-RNCRN
displays the desired behavior which, owing to 
effective replenishment of the fuel species, is now
maintained over the desired interval of time.

\section{Appendix: Robustness}
\label{sec:robustness}
In this section, we provide a preliminary investigation of robustness
of some RNCRNs approximating the oscillatory target system~(\ref{eq:bessel_osc}).
In particular, we first consider the  
RNCRN with RREs~(\ref{eq:full_rncrn_oscil}).
This RNCRN was not optimized for robustness;
nevertheless, we show that it is not 
pathologically sensitive to perturbations.
We then apply a noisy version of 
Algorithm~\ref{algh:RNCRN} on the target system~(\ref{eq:bessel_osc}) to produce a new oscillatory RNCRN which is 
more robust to perturbations. 

\subsection{Rate-coefficient robustness}
\label{sec:rate_robustness}
Let us consider RREs~(\ref{eq:full_rncrn_oscil}) 
with perturbed rate coefficients and fixed 
perceptron speed $\mu$, given by
\begin{align}
\frac{\mathrm{d} \hat{x}_i}{\mathrm{d} t} & = \hat{\beta}_{i} +  \hat{x}_i \sum_{j=1}^M \hat{\alpha}_{i,j} \hat{y}_j,  
&& \hat{x}_i(0) = a_i, 
\; \; \; \textrm{for } i = 1, 2, \nonumber \\
\frac{\mathrm{d} \hat{y}_j}{\mathrm{d} t} & =  
\frac{\hat{\gamma}}{\mu} 
+ \frac{\hat{\theta}_{j}}{\mu} \hat{y}_j 
+ \left(\sum_{i=1}^N \frac{\hat{\omega}_{j,i}}{\mu} \hat{x}_i \right) y_j - \frac{\hat{\kappa}_j}{\mu} \hat{y}_j^2, 
&& \hat{y}_j(0) = b_j, 
\; \; \; \textrm{for } j = 1, 2, \ldots, 6,
\label{eq:single_layer_RRE_noise}
\end{align}
with coefficients
\begin{align}
\hat{\beta}_{i} & =  1 + \eta r_{\beta_i}, 
\; \; 
\hat{\alpha}_{i,j} =  \left(1 + \eta r_{\alpha_{i,j}} \right) 
\alpha_{i,j},
\; \; 
\hat{\gamma} = 1 + \eta r_{\gamma}, \nonumber \\
\hat{\omega}_{j,i} & = \left(1 + \eta r_{\omega_{j,i}} \right) \omega_{j,i}, 
\; \; 
\hat{\theta}_{j} = 
\left(1 + \eta r_{\theta_j} \right) \theta_j,
\; \; 
\hat{\kappa}_{i} =  1 + \eta r_{\kappa_i}, 
\label{eq:single_layer_RRE_noise_coeff}
\end{align}
where $r \in (-1,1)$ with appropriate subscripts are
independent uniformly distributed random variables, 
scaled with a parameter $\eta \in [0,1]$, 
i.e. $\eta r$ is uniformly distributed in $[-\eta,\eta]$.

To study robustness, we fix the initial conditions to
$a_1 = 2$, $a_2 = 4$ and $b_j = 0$,
the perceptron speed to $\mu = 0.1$, 
and the noise strength to $\eta \in [0,1]$. 
Then, for a fixed sampling of coefficients~(\ref{eq:single_layer_RRE_noise_coeff}), 
we numerically solve the resulting perturbed RREs~(\ref{eq:single_layer_RRE_noise}).
We then verify if the solution $\hat{x}_1(t)$ 
is a periodic function of time 
within the time-interval $t \in [0,60]$; if this is the case, 
then we say that RREs~(\ref{eq:full_rncrn_oscil})
are robust with respect to this particular set of perturbations.
For this verification,  we use a fast Fourier transform 
and check for the existence of suitable peaks, above a threshold, in the frequency domain (see code~\cite{Dack_github_2024} for implementation).
We then repeat this procedure for sufficiently many samplings 
of~(\ref{eq:single_layer_RRE_noise_coeff}), and compute
the ratio of oscillatory to non-oscillatory perturbed systems, 
which we use as an estimate for rate-coefficent robustness
for fixed $\eta$. We then repeat the same computations for
various values of noise strength $\eta \in [0,1]$.
See Figure~\ref{fig:bessel_osc_traj_cherr_net_robustness}(a)--(b).

\subsection{Perceptron-speed robustness}
\label{app:perceptron_speed_robustness}
To study robustness with respect to the
parameter $\mu$, we numerically solve the RREs~(\ref{eq:full_rncrn_oscil}) 
with rate coefficients~(\ref{eq:coefficents_oscill})
and various values of $\mu$.
The results are shown in Figure~\ref{fig:bessel_osc_traj_cherr_net_robustness}(c)--(d).

\subsection{Initial-condition robustness}
\label{app:init_conc_robustness}
Let us now consider RREs~(\ref{eq:full_rncrn_oscil}) with perturbed initial conditions, given by
\begin{align}
\frac{\mathrm{d} \hat{x}_i}{\mathrm{d} t} & = \beta_{i} +  \hat{x}_i \sum_{j=1}^M \alpha_{i,j} \hat{y}_j,  
&& \hat{x}_i(0) = (1 + \eta r_{a_i}) a_i, 
\; \; \textrm{for } i = 1, 2, \nonumber \\
\frac{\mathrm{d} \hat{y}_j}{\mathrm{d} t} & =  
\frac{\gamma}{\mu} 
+ \frac{\theta_{j}}{\mu} \hat{y}_j 
+ \left(\sum_{i=1}^N \frac{\omega_{j,i}}{\mu} \hat{x}_i \right) y_j - \frac{\kappa_j}{\mu} \hat{y}_j^2,  
&& \hat{y}_j(0) = (1 + \eta r_{b_j}) b_j, 
\; \; \textrm{for } j = 1, 2, \ldots, 6,
\label{eq:single_layer_RRE_noise_ICs}
\end{align}
where $r \in (-1,1)$ with appropriate subscripts are
independent uniformly distributed random variables, 
and parameter $\eta \in [0,1]$.

We fix the rate coefficients from~(\ref{eq:single_layer_RRE_noise_ICs})
to~(\ref{eq:coefficents_oscill}) and $\mu = 0.1$.
Furthermore, we fix $a_1 = 2$ and $a_2=4$, 
while $b_j$ are fixed to the equilibrium values of
the chemical perceptrons:
$b_1 = 1.79$,  $b_2 = 0.17$,   $b_3 = 0.20$, 
$b_4 = 4.19$, $b_5 = 0.42$, and $b_6 = 1.22$.
The results are presented in Figure~\ref{fig:bessel_osc_traj_cherr_net_robustness}(e)--(f).

\subsection{Noisy training}
\label{app:noisy_training}
The RNCRN with RREs~(\ref{eq:full_rncrn_oscil})
has been produced by applying Algorithm~\ref{algh:RNCRN}
with deterministic coefficients 
$\beta_i, \alpha_{i,j}, \gamma, \omega_{j,i}, \theta_j$
in~(\ref{eq:single_layer_RRE_reduced}). 
In order to design an RNCRN which is more robust
to perturbations in the rate coefficients, 
one may wish to instead apply Algorithm~\ref{algh:RNCRN}
with coefficients $\beta_i, \alpha_{i,j}, \gamma, \omega_{j,i}, \theta_j$
being stochastically perturbed 
at each iteration of the backpropagation algorithm.
For the purpose of this paper, instead of 
perturbing each of the coefficients, we introduce 
only three perturbations, which we choose 
in a way natural to our computer code. 
In particular, let us now consider Algorithm~\ref{algh:RNCRN} 
with one modification: in step (a),
we ensure that the distance between 
$(f_i(x_1,x_2,\ldots,x_N)/x_i - \beta_i/x_i)$ and 
$\sum_{j=1}^M \alpha_{i,j}^* 
\left( \delta^{(y)}_j + \sigma_{\gamma} \left(\sum_{k=1}^N 
\omega_{j,k}^* (x_k+\delta_k^{(x)}) + \theta_{j}^* \right) \right) + \delta_i^{(\beta)}$ is within the tolerance, 
instead of $\sum_{j=1}^M \alpha_{i,j}^* 
\sigma_{\gamma} \left(\sum_{k=1}^N 
\omega_{j,k}^* x_k + \theta_{j}^* \right)$.
Here, $\delta_k^{(x)}$, $\delta^{(y)}_j$, $\delta_i^{(\beta)}$
are independent normal random variables with zero mean and variance
$\eta^2$, which we prescribe. These three random variables
are resampled at each iteration of the backpropagation algorithm.
See the code presented in the folder entitled 
``step\_4\_noisy\_training'' within the folder  ``ex\_6\_robustness'' 
available in~\cite{Dack_github_2024}.

\begin{figure}[H]
\centering
\includegraphics[width=\columnwidth]{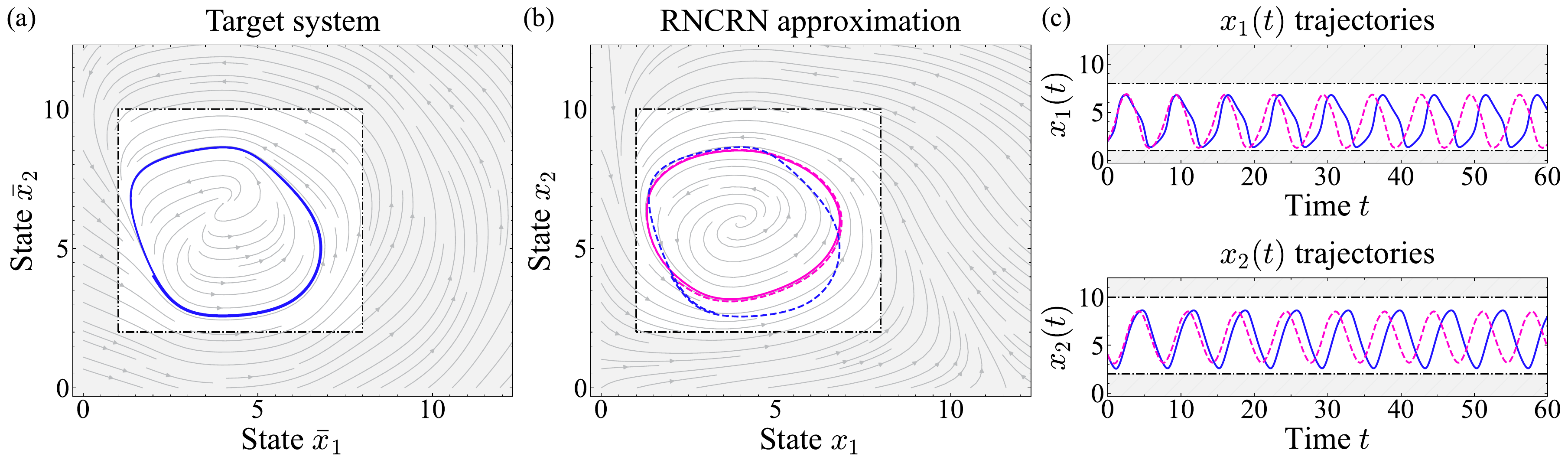}
\caption{RNCRN approximation of the oscillatory target system~(\ref{eq:bessel_osc}) trained with noisy Algorithm~\ref{algh:RNCRN}. 
(a) The vector field of the target system~(\ref{eq:bessel_osc}) (grey arrows) for $\mathbb{K}_1
\times \mathbb{K}_2 = [1,8] \times [2, 10]$, and a solution of~(\ref{eq:bessel_osc}) with $\bar{x}_1(0)=2$ and $\bar{x}_2(0)=4$ (blue). (b) Analogous plot is shown for the RNCRN whose reduced and full ODEs are given respectively by~(\ref{eq:single_layer_RRE_reduced_oscill_noise}) 
and~(\ref{eq:full_rncrn_oscil_noise}), with 
coefficients~(\ref{eq:coefficents_oscill_noise}) and $\mu = 0.01$. 
In particular, displayed as grey arrows is the vector field
of the reduced ODEs~(\ref{eq:single_layer_RRE_reduced_oscill_noise}), 
together with an oscillatory solution of the full ODEs~(\ref{eq:full_rncrn_oscil_noise}) shown in purple; 
for comparison, we also display as dashed blue curve 
the oscillatory solution of~(\ref{eq:bessel_osc}) from (a).
Panel (c) displays the solutions $\bar{x}_1(t)$ and $x_1(t)$, 
and $\bar{x}_2(t)$ and $x_2(t)$. The initial concentrations of the perceptrons are all set to zero.}
\label{fig:bessel_noise_trained_RNCRN}
\end{figure}

Applying this noisy version of Algorithm~\ref{algh:RNCRN} to the target system~(\ref{eq:bessel_osc}) on $\mathbb{K}_1 \times \mathbb{K}_2 = [1,8] \times [2, 10]$, 
with tolerance $\varepsilon \approx 10^{-1}$ and variance $\eta^2=1$, 
we find an RNCRN with $M=6$ perceptrons, coefficients $\beta_1=\beta_2 = 1$, $\gamma = 1$, and rounded to 3 decimal places:
\begin{align}
    \boldsymbol{\alpha}_1 &= \begin{pmatrix}
-0.034\\
-0.040\\
0.017\\
0.036 \\
0.037 \\
-0.011
    \end{pmatrix}, \; \; 
    \boldsymbol{\alpha}_2 = \begin{pmatrix} 
0 \\
-0.014\\
-0.023 \\
-0.013 \\
0.011 \\
0.020 
    \end{pmatrix}, \;\;
    \boldsymbol{\theta} &= \begin{pmatrix}
    -18.748 \\
    34.068 \\	
    39.575\\
    36.980 \\
    49.304\\	
    -1.085
    \end{pmatrix}, \;\;
    \boldsymbol{\omega}_1 = \begin{pmatrix}
    -2.101 \\
    -7.724 \\
    -5.370 \\
    -1.695 \\
    -13.497 \\
    5.802
    \end{pmatrix}, \;\;
    \boldsymbol{\omega}_2 = \begin{pmatrix}
        5.662 \\
    0.522 \\
    -4.331 \\
    -3.996 \\
    -3.471 \\	
    -3.901
    \end{pmatrix},
    \label{eq:coefficents_oscill_noise}
\end{align}
where $\boldsymbol{\alpha}_i = (\alpha_{i,1},\alpha_{i,2},\ldots, \alpha_{i,6})^{\top}$ for $i = 1,2$, 
$\boldsymbol{\theta} = 
(\theta_{1},\theta_{2},\theta_{3},\theta_{4},\theta_{5}, \theta_{6})^{\top}$, and  $\boldsymbol{\omega}_k = 
(\omega_{1,k},\omega_{2,k},\ldots,\omega_{6,k})^{\top}$
for $k = 1, 2$.
The reduced ODEs are given by 
\begin{align}
\frac{\mathrm{d} \tilde{x}_1}{\mathrm{d} t} & = 
g_1(\tilde{x}_1, \tilde{x}_2) = 1 + \tilde{x}_1 \sum_{j=1}^{6}\alpha_{1,j} 
\sigma_{1} \left(\omega_{j,1} \tilde{x}_1 + \omega_{j,2} \tilde{x}_2 + \theta_{j} \right), \nonumber \\
\frac{\mathrm{d} \tilde{x}_2}{\mathrm{d} t} & = 
g_2(\tilde{x}_1, \tilde{x}_2) = 1 + \tilde{x}_2 \sum_{j=1}^{6}\alpha_{2,j} 
\sigma_{1} \left(\omega_{j,1} \tilde{x}_1 + \omega_{j,2} \tilde{x}_2 + \theta_{j} \right), 
\label{eq:single_layer_RRE_reduced_oscill_noise}
\end{align}
 while the full ODEs read
\begin{align}
\frac{\mathrm{d} x_1}{\mathrm{d}t} 
& = 1+ x_1\left(\sum_{j=1}^6 \alpha_{1,j}y_j\right), 
&& \hat{x}_1(0) = a_1 \in \mathbb{K}_1, \nonumber \\
\frac{\mathrm{d} x_2}{\mathrm{d}t} 
& = 1+ x_2\left(\sum_{j=1}^6 \alpha_{2,j}y_j\right), 
&& \hat{x}_2(0) = a_2 \in \mathbb{K}_2, \nonumber \\
\frac{\mathrm{d} y_j}{\mathrm{d}t} 
&= \frac{1}{\mu} + \frac{\theta_{j}}{\mu} y_j 
+\left( \sum_{i=1}^2 \frac{\omega_{j,i}}{\mu} x_i \right) y_j 
- \frac{1}{\mu} y_j^2, 
&& y_j(0)= b_j \geq 0, 
\; \; \; j = 1,2,\ldots,6.
\label{eq:full_rncrn_oscil_noise}
\end{align}
We present the dynamics of~(\ref{eq:full_rncrn_oscil_noise}) in Figure~\ref{fig:bessel_noise_trained_RNCRN},
which is analogous to Figure~\ref{fig:bessel_osc_traj_cherr_net}.
Furthermore, in Figure~\ref{fig:bessel_osc_traj_cherr_net_robustness}(a), 
we show that this RNCRN displays superior robustness.

\section{Appendix: Alternative approximation scheme}
\label{sec:compare}
In this section, we outline an alternative method
for approximating non-polynomial ODEs with RREs~\cite{kerner_universal_1981,kowalski_universal_1993}, 
and point out some of its drawbacks.

\textbf{Non-polynomial to polynomial ODEs}.
Consider again the non-polynomial ODE~(\ref{eq:multistable_ex}).
To map it to a polynomial one, we apply the method from~\cite{kerner_universal_1981}: 
we introduce the auxiliary variables
$y_1 = \sin{(x_1)}$ and $y_2 = \cos{(x_1)}$, and use
the chain-rule to obtain
\begin{align}
\frac{\mathrm{d} x_1}{\mathrm{d}t} &= y_1, 
&&x_1(0) = a_1,\nonumber\\
\frac{\mathrm{d} y_1}{\mathrm{d}t} &= y_1y_2, 
&&y_1(0) = \sin(a_1), \nonumber \\
\frac{\mathrm{d} y_2}{\mathrm{d}t} &= -y_1^2, 
&&y_2(0) = \cos(a_1).
\label{eq:polynomial_neg_init}
\end{align}

\textbf{Polynomial ODEs to RREs}. 
To ensure that the initial conditions for $y_1$ and $y_2$
are non-negative, we translate the variables according to
$z_1(t) = T_1 + y_1(t)$ and $z_2(t) = T_2 + y_2(t)$, 
where $T_1, T_2 >1$; for concreteness, we fix $T_1 = T_2 = 2$.
Then, (\ref{eq:polynomial_neg_init}) becomes
\begin{align}
\frac{\mathrm{d} x_1}{\mathrm{d}t} 
& = - 2 + z_1, 
&&x_1(0) = a_1,\nonumber\\
\frac{\mathrm{d} z_1}{\mathrm{d}t} 
& = 4 - 2 z_1 - 2 z_2 + z_1 z_2, 
&&z_1(0) = 2 + \sin(a_1), \nonumber \\
\frac{\mathrm{d} z_2}{\mathrm{d}t} 
& = - 4 + 4 z_1 -z_1^2, 
&&z_2(0) = 2 + \cos(a_1).
\label{eq:polynomial_pos_init}
\end{align}
Terms $-2$, $-2 z_2$ and $- 4 -z_1^2$
in respectively the first, second and third equation 
of~(\ref{eq:polynomial_pos_init}) cannot be 
interpreted as chemical reactions~\cite{plesa_chemical_2016}.
Multiple methods exist to eliminate such non-chemical terms~\cite{plesa_chemical_2016,samardzija_nonlinear_1989,poland_cooperative_1993,  holcman_test_2017,plesa_integral_2021, hangos_mass_2011};
in what follows, we use incomplete 
Carleman embedding~\cite{kowalski_universal_1993}.
In particular, introducing additional variables
 $z_3 = 1/x_1$, $z_4 = 1/z_1$, and $z_5 = 1/z_2$, 
 and applying the chain-rule, yields
\begin{align}
\frac{\mathrm{d} x_1}{\mathrm{d}t} 
& = - 2 x_1 z_3 + z_1, 
&& x_1(0) = a_1,\nonumber\\
\frac{\mathrm{d} z_1}{\mathrm{d}t} 
& = 4 - 2 z_1 - 2 z_1 z_2 z_4 + z_1 z_2, 
&& z_1(0) = 2 + \sin(a_1), \nonumber \\
\frac{\mathrm{d} z_2}{\mathrm{d}t} 
& = - 4 z_2 z_5 + 4 z_1 - z_1^2 z_2 z_5, 
&& z_2(0) = 2 + \cos(a_1), \nonumber \\
\frac{\mathrm{d} z_3}{\mathrm{d}t} 
& = 2 z_3^2 - z_1 z_3^2,
&& z_3(0) = \frac{1}{a_1},\nonumber\\
\frac{\mathrm{d} z_4}{\mathrm{d}t} 
& =  2 z_4 - 4 z_4^2 - z_2 z_4 + 2 z_2 z_4^2, 
&& z_4(0) = \frac{1}{2 + \sin(a_1)}, \nonumber \\
\frac{\mathrm{d} z_5}{\mathrm{d}t} 
& = 4 z_5^2 - 4 z_1 z_5^2 +z_1^2 z_5^2, 
&& z_5(0) = \frac{1}{2 + \cos(a_1)}. 
\label{eq:RRE_analytic_high_order}
\end{align}
RREs~(\ref{eq:RRE_analytic_high_order}) contain
terms of degree as high as four; for example, the term
$- z_1^2 z_2 z_5$. This term induces the reaction 
$2Z_1 + Z_2 + Z_5 \xrightarrow{1} 2Z_1 + Z_5$, 
which has four reactants, and is therefore 
experimentally infeasible. 
One could approximate such higher-order reactions
with systems of second-order ones~\cite{wilhelm_chemical_2000,plesa_stochastic_2023}.
However, for the purpose of this section,  
RREs~(\ref{eq:RRE_analytic_high_order}) are sufficient.

\textbf{Initial-condition robustness}.
Let us now numerically study robustness of~(\ref{eq:RRE_analytic_high_order}) with respect 
to the initial conditions.
In Figure~\ref{fig:carleman_approx}(a), 
we display the solution $x_1(t)$ of~(\ref{eq:RRE_analytic_high_order}),
i.e. with ideal initial conditions. 
One can notice that $x_1(t)$ from~(\ref{eq:RRE_analytic_high_order}) 
is \emph{identical} to $\bar{x}_1(t)$
from the target system~(\ref{eq:multistable_ex}).
However, note that this perfect accuracy 
requires one to perfectly adjust the initial conditions
for the auxiliary species $z_i(0)$ for any given 
initial condition $x_i(0)$, i.e. if one changes $x_i(0)$, 
then one also has to appropriately 
change $z_i(0)$ for all $i = 1,2,\ldots, 5$.
Critically, as we now show, this perfect accuracy
can deteriorate catastrophically under some perturbations 
of the initial conditions. 

In particular, let us consider~(\ref{eq:RRE_analytic_high_order})
with perturbed initial conditions:
\begin{align}
\frac{\mathrm{d} \hat{x}_1}{\mathrm{d}t} 
& = - 2 \hat{x}_1 \hat{z}_3 + \hat{z}_1, 
&& \hat{x}_1(0) = (1+\eta r_{1}) a_1,\nonumber\\
\frac{\mathrm{d} \hat{z}_1}{\mathrm{d}t} 
& = 4 - 2 \hat{z}_1 - 2 \hat{z}_1 \hat{z}_2 \hat{z}_4 
+ \hat{z}_1 \hat{z}_2, 
&& \hat{z}_1(0) = (1 + \eta r_{2})(2 + \sin(a_1)), \nonumber \\
\frac{\mathrm{d} \hat{z}_2}{\mathrm{d}t} 
& = - 4 \hat{z}_2 \hat{z}_5 + 4 \hat{z}_1
- \hat{z}_1^2 \hat{z}_2 \hat{z}_5, 
&& \hat{z}_2(0) = (1 + \eta r_{3}) (2 + \cos(a_1)), \nonumber \\
\frac{\mathrm{d} \hat{z}_3}{\mathrm{d}t} 
& = 2 \hat{z}_3^2 - \hat{z}_1 \hat{z}_3^2,
&& \hat{z}_3(0) = (1 + \eta r_{4}) \frac{1}{a_1},\nonumber\\
\frac{\mathrm{d} \hat{z}_4}{\mathrm{d}t} 
& =  2 \hat{z}_4 - 4 \hat{z}_4^2 - \hat{z}_2 \hat{z}_4 
+ 2 \hat{z}_2 \hat{z}_4^2, 
&& \hat{z}_4(0) = (1 + \eta r_{5}) \frac{1}{2 + \sin(a_1)}, \nonumber \\
\frac{\mathrm{d} \hat{z}_5}{\mathrm{d}t} 
& = 4 \hat{z}_5^2 - 4 \hat{z}_1 \hat{z}_5^2 
+ \hat{z}_1^2 \hat{z}_5^2, 
&& \hat{z}_5(0) = (1 + \eta r_{6}) \frac{1}{2 + \cos(a_1)}, 
\label{eq:RRE_analytic_high_order_perturbed}
\end{align}
where, as in Appendix~\ref{app:init_conc_robustness}, 
$r_i \in (-1,1)$ are independent uniformly 
distributed random variables, 
and parameter $\eta \in [0,1]$.
In the top panel of Figure~\ref{fig:carleman_approx}(b), we show 
as a function of $\eta$
the proportion of perturbed systems~(\ref{eq:RRE_analytic_high_order_perturbed}),
over a range of initial conditions, whose 
solutions $x_1(t)$ 
converge to within $\pm \pi/4$ of the intended equilibrium; 
one can notice an increasing
fragility beyond $\eta = 10\%$. Critically, 
for some non-ideal initial conditions, concentrations of some
of the auxiliary species grow unboundedly (blow-up),
as shown in the bottom panel of Figure~\ref{fig:carleman_approx}(b).
Let us note that no such chemically hazardous behavior
is observed in the RNCRN. In Figure~\ref{fig:carleman_approx}(c), 
we display a particular data point from Figure~\ref{fig:carleman_approx}(b), showing that $x_1(t)$ from~(\ref{eq:RRE_analytic_high_order_perturbed}) can fail to converge
close to the intended equilibrium $\pi$, 
while, at the same time, the auxiliary species concentration
$z_5(t)$ blows up.

\begin{figure}[ht]
\centering
\includegraphics[width=\columnwidth]{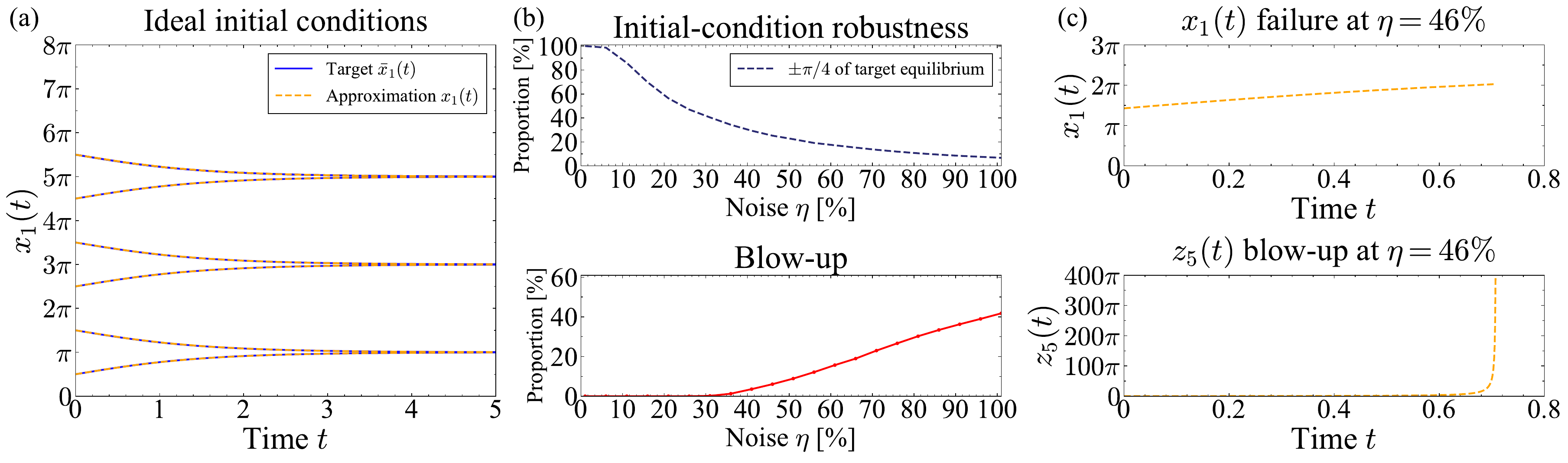}
\caption{Initial-condition sensitivity of system~(\ref{eq:RRE_analytic_high_order}) that approximates
target system~(\ref{eq:multistable_ex}).
Panel (a) shows solution $\bar{x}_1(t)$ of~(\ref{eq:multistable_ex}), and $x_1(t)$ of~(\ref{eq:RRE_analytic_high_order}). 
Top sub-panel (b) shows the proportion of perturbed systems~(\ref{eq:RRE_analytic_high_order_perturbed}), 
over a range of initial conditions,
whose solutions $x_1(t)$ are within $\pm \pi/4$ 
of the desired equilibrium at $t=6$ 
as a function of noise strength $\eta$. 
Bottom sub-panel (b) shows the proportion of~(\ref{eq:RRE_analytic_high_order_perturbed})
with at least one auxiliary species that blows up
before $t = 6$.
Each point in panel (b) was obtained by simulating~(\ref{eq:RRE_analytic_high_order_perturbed})
with $10^4$ sampled initial conditions.  
Panel (c) shows an example of failure of $x_1(t)$
to converge close to the desired equilibrium,
and a blow-up of an auxiliary species. 
The initial conditions are as follows:
 $x_1(t) = 4.48121513$, $z_1(0) =  4.13072385$, $z_2(0) = 1.55830895$, $z_3 (0) = 0.08109365$, $z_4(0) =  0.22538005$, $z_5(0) = 0.56915378$.}
\label{fig:carleman_approx}
\end{figure}

\section{Appendix: Algorithm~\ref{algh:RNCRN} pseudocode}
\label{sec:pseudocode} 

\begin{algorithm}
\caption{Two-step algorithm for training the RNCRN}\label{alg:cap}
\begin{algorithmic}
\Require $f_i( \bar{x}_1, \dots, \bar{x}_N )$ for $i= 1,2, \dots, N$ \Comment{target system~(\ref{eq:target_ODEs})}
\Require $\mathbb{K}_1, \mathbb{K}_2, \ldots, \mathbb{K}_N \subset (0,+\infty)$ \Comment{target compact sets}
\Require $\gamma > 0$ and $\beta_1,\beta_2,\ldots,\beta_N \ge 0$ \Comment{positive and nonnegative rate coefficients}
\Require $\varepsilon_Q > 0$ and $\varepsilon_D > 0$ \Comment{quasi-static and dynamical tolerance}
\Require $\mathcal{L}_Q$ and $\mathcal{L}_D$ \Comment{loss functions for the quasi-static and the dynamic approximation}
\Require $T>0$ \Comment{time length of dynamics}
\Require $a_1,a_2,\ldots,a_M, b_1,b_2,\ldots,b_M \ge 0 $ \Comment{initial conditions of target system }
\Require{$ 0 < \mu \ll 1$} \Comment{speed of chemical perceptrons lower bound}
\Require \Call{Integrate}{$\bold{f}, t \in [0, T], \bold{z}(0)$} \Comment{procedure for numerical integration}
\Require  \Call{Minimize}{$\mathcal{L}$, $\Omega$} \Comment{procedure to minimize loss (\textit{i.e.} backpropagation~\cite{rumelhart_learning_1986} )}
\Require  \Call{Reduce}{$\mu$} \Comment{procedure to reduce $\mu^*$ and $\varepsilon_Q$ }
\Procedure{Quasi-static approximation}{$\varepsilon_Q$}
\State $M \gets 0$
\State $\varepsilon_Q^* \gets \infty$
\While{$\varepsilon_Q^* > \varepsilon_Q$} 
\State $M \gets M+1$
\State $\mathcal{L} \gets \mathcal{L}_Q \left \{ f_i(x_1,x_2,\ldots,x_N)/x_i - \beta_i/x_i), 
\sum_{j=1}^M \alpha_{i,j} 
\sigma_{\gamma} \left(\sum_{k=1}^N 
\omega_{j,k} x_k + \theta_{j} \right) \right \}$
\State $\varepsilon_Q^*, \alpha_{i,j}^*, \omega_{j,k}^*, \theta_{j}^*  \gets$ \Call{Minimize}{$\mathcal{L}, (\alpha_{i,j}, \omega_{j,k}, \theta_{j})$}
\EndWhile
\State \textbf{return} $\varepsilon_Q^*, \alpha_{i,j}^*$, $\omega_{j,k}^*$, $\theta_{j}^*$
\EndProcedure
\Procedure{Dynamical approximation}{$\varepsilon_Q$}
\State $\varepsilon_D^* \gets \infty$
\State $\mu^* \gets 1$
\State $\bold{\bar{z}}(0) \gets (a_1,a_2,\ldots,a_M)$
\State $\bold{z}(0) \gets (a_1,a_2,\ldots,a_M, b_1,b_2,\ldots,b_M)$
\State $\bold{TARGET} \gets $ Equation~(\ref{eq:target_ODEs})
\State $\bar{x}_1(t), \dots, \bar{x}_N(t) \gets$\Call{Integrate}{$\bold{\bold{TARGET}}, t \in [0, T], \bold{\bar{z}}(0)$}
\State  $\varepsilon_Q^*, \alpha_{i,j}^*$, $\omega_{j,k}^*$, $\theta_{j}^* \gets$ \Call{Quasi-static approximation}{$\varepsilon_Q$}

\While{$\varepsilon_D^* > \varepsilon_D$ and $\mu^* > \mu$}
\State $\mu^* \gets$ \Call{Reduce}{$\mu^*$} 
\State $\bold{RNCRN} \gets $ Equation~(\ref{eq:single_layer_RRE}) with $\alpha_{i,j} = \alpha_{i,j}^*$,
$\theta_{j} = \theta_{j}^*$, $\omega_{j,i} = \omega_{j,i}^*$, $\gamma$, $\mu^*$, and $\beta_1,\beta_2,\ldots,\beta_N$

\State $x_1(t), \dots, x_N(t), y_1(t), \dots, y_M(t) \gets$ \Call{Integrate}{$\bold{RNCRN}, t \in [0, T], \bold{z}(0)$}
\State $\varepsilon_D^* \gets \mathcal{L}_D\{ x_i(t), \bar{x}_i(t)\}$ over all $t \in [0, T]$ and all $i = 1, 2, \dots, N$
\EndWhile
\If{$\varepsilon_D^* < \varepsilon_D$}
\State \textbf{return} $\alpha_{i,j}^*$, $\omega_{j,k}^*$, $\theta_{j}^*$, $\mu^*$
\Else
\State $\varepsilon_Q \gets $ \Call{Reduce}{$\varepsilon_Q$}
\State \textbf{return} \Call{Dynamical approximation}{$\varepsilon_Q$} \Comment{recursive procedure} 
\EndIf
\EndProcedure
%\State $\varepsilon^*, \alpha_{i,j}^*$, $\omega_{j,k}^*$, $\omega_{j}^*\gets $\Call{Quasi-static approximation}{}
%\If{\Call{Dynamical approximation}{$\mu$}}
\end{algorithmic}
\end{algorithm}

\end{document}